\def\year{2022}\relax
\documentclass[letterpaper]{article} %
\makeatletter
\def\year{2022}
 \newlength\titlebox 
 \twocolumn 
\def\addcontentsline#1#2#3{}
\def\copyright@year{\number\year}
\def\copyright@text{Copyright \copyright\space \copyright@year,
Association for the Advancement of Artificial Intelligence (www.aaai.org).
All rights reserved.}
\def\copyright@on{T}
\def\nocopyright{\gdef\copyright@on{}} %
\def\copyrighttext#1{\gdef\copyright@on{T}\gdef\copyright@text{#1}}
\def\copyrightyear#1{\gdef\copyright@on{T}\gdef\copyright@year{#1}}
\def\maketitle{\par
\begingroup %
\def\thefootnote{\fnsymbol{footnote}}
\twocolumn[\@maketitle] \@thanks
\endgroup
\if T\copyright@on\insert\footins{\noindent\footnotesize\copyright@text}\fi
\setcounter{footnote}{0}
\let\maketitle\relax \let\@maketitle\relax
\gdef\@thanks{}\gdef\@author{}\gdef\@title{}\let\thanks\relax}
\long\gdef\affiliations #1{ \def \affiliations_{#1}}

\def\@maketitle{%
  \newcounter{eqfn}\setcounter{eqfn}{0}%
  \newsavebox{\titlearea}
  \sbox{\titlearea}{
    \let\footnote\relax\let\thanks\relax%
    \setcounter{footnote}{0}%
    \def\equalcontrib{%
      \ifnum\value{eqfn}=0%
        \footnote{These authors contributed equally.}%
        \setcounter{eqfn}{\value{footnote}}%
      \else%
        \footnotemark[\value{eqfn}]%
      \fi%
    }%
    \vbox{%
      \hsize\textwidth%
      \linewidth\hsize%
      \vskip 0.625in minus 0.125in%
      \centering%
      {\LARGE\bf \@title \par}%
      \vskip 0.1in plus 0.5fil minus 0.05in%
      {\Large{\textbf{\@author\ifhmode\\\fi}}}%
      \vskip .2em plus 0.25fil%
      {\normalsize \affiliations_\ifhmode\\\fi}%
      \vskip .5em plus 2fil%
    }%
  }%
  \newlength\actualheight%
  \settoheight{\actualheight}{\usebox{\titlearea}}%
  \ifdim\actualheight>\titlebox%
    \setlength{\titlebox}{\actualheight}%
  \fi%
  \vbox to \titlebox {%
    \let\footnote\thanks\relax%
    \setcounter{footnote}{0}%
    \def\equalcontrib{%
      \ifnum\value{eqfn}=0%
        \footnote{These authors contributed equally.}%
        \setcounter{eqfn}{\value{footnote}}%
      \else%
        \footnotemark[\value{eqfn}]%
      \fi%
    }%
    \hsize\textwidth%
    \linewidth\hsize%
    \vskip 0.625in minus 0.125in%
    \centering%
    {\LARGE\bf \@title \par}%
    \vskip 0.1in plus 0.5fil minus 0.05in%
    {\Large{\textbf{\@author\ifhmode\\\fi}}}%
    \vskip .2em plus 0.25fil%
    {\normalsize \affiliations_\ifhmode\\\fi}%
    \vskip .5em plus 2fil%
  }%
}%
\renewenvironment{abstract}{%
  \centerline{\bf Abstract}%
  \vspace{0.5ex}%
  \setlength{\leftmargini}{10pt}%
  \begin{quote}%
    \small%
}{%
  \par%
  \end{quote}%
  \vskip 1ex%
}%
\def\section{\@startsection {section}{1}{\z@}{-2.0ex plus
-0.5ex minus -.2ex}{3pt plus 2pt minus 1pt}{\Large\bf\centering}}
\def\subsection{\@startsection{subsection}{2}{\z@}{-2.0ex plus
-0.5ex minus -.2ex}{3pt plus 2pt minus 1pt}{\large\bf\raggedright}}
\def\subsubsection{\@startsection{subparagraph}{3}{\z@}{-6pt plus
-2pt minus -1pt}{-1em}{\normalsize\bf}}
\renewcommand\paragraph{\@startsection{paragraph}{4}{\z@}{-6pt plus -2pt minus -1pt}{-1em}{\normalsize\bf}}%
\skip\footins 9pt plus 4pt minus 2pt
\def\footnoterule{\kern-3pt \hrule width 5pc \kern 2.6pt }
\labelwidth\leftmargini\advance\labelwidth-\labelsep \labelsep 5pt
\def\@listi{\leftmargin\leftmargini}
\def\@listii{\leftmargin\leftmarginii
\labelwidth\leftmarginii\advance\labelwidth-\labelsep
\topsep 2pt plus 1pt minus 0.5pt
\parsep 1pt plus 0.5pt minus 0.5pt
\itemsep \parsep}
\def\@listiii{\leftmargin\leftmarginiii
\labelwidth\leftmarginiii\advance\labelwidth-\labelsep
\topsep 1pt plus 0.5pt minus 0.5pt
\parsep \z@
\partopsep 0.5pt plus 0pt minus 0.5pt
\itemsep \topsep}
\def\@listiv{\leftmargin\leftmarginiv
\labelwidth\leftmarginiv\advance\labelwidth-\labelsep}
\def\@listv{\leftmargin\leftmarginv
\labelwidth\leftmarginv\advance\labelwidth-\labelsep}
\def\@listvi{\leftmargin\leftmarginvi
\labelwidth\leftmarginvi\advance\labelwidth-\labelsep}
\belowdisplayskip \abovedisplayskip
\def\normalsize{\@setfontsize\normalsize\@xpt{11}}   %
\def\small{\@setfontsize\small\@ixpt{10}}    %
\def\footnotesize{\@setfontsize\footnotesize\@ixpt{10}}  %
\def\scriptsize{\@setfontsize\scriptsize\@viipt{10}}  %
\def\tiny{\@setfontsize\tiny\@vipt{7}}    %
\def\large{\@setfontsize\large\@xipt{12}}    %
\def\Large{\@setfontsize\Large\@xiipt{14}}    %
\def\LARGE{\@setfontsize\LARGE\@xivpt{16}}    %
\def\huge{\@setfontsize\huge\@xviipt{20}}    %
\def\Huge{\@setfontsize\Huge\@xxpt{23}}    %

\AtBeginDocument{\@ifpackageloaded{natbib}
  {
    \let\cite\citep
    \let\shortcite\citeyearpar
    \setcitestyle{aysep={}}
    \setlength\bibhang{0pt}
    \bibliographystyle{aaai22}
  }{}
}

\makeatother
\usepackage{times}  %
\usepackage[scaled=0.90]{helvet}  %
\usepackage{courier}  %
\usepackage[hyphens]{url}  %
\usepackage{graphicx} %
\usepackage{natbib}  %
\usepackage{caption} %
\DeclareCaptionStyle{ruled}{labelfont=normalfont,labelsep=colon,strut=off} %
\usepackage{algorithm}
\usepackage{algorithmic}

\usepackage{newfloat}
\usepackage{listings}
\floatstyle{ruled}
\newfloat{listing}{tb}{lst}{}
\floatname{listing}{Listing}
\usepackage{amsmath}
\usepackage{amssymb}
\usepackage{amsfonts}
\usepackage{amsthm}
\RequirePackage{xspace}

\usepackage{todonotes}

\newcommand{\abbExtraSpace}{}
\newcommand\eg{e.\abbExtraSpace g.,\xspace}
\newcommand\ie{i.\abbExtraSpace e.,\xspace}
\newcommand\cf{cf.\xspace}

\newcommand\wrt{w.\abbExtraSpace r.\abbExtraSpace t.\xspace}

\newcommand{\opnotation}[1]{\text{{\sf #1}}}

\renewcommand{\vec}[1]{\boldsymbol{#1}}
\newcommand{\tuple}[1]{\langle{#1}\rangle}
\newcommand\arity{\opnotation{ar}}

\newcommand\range[1]{\opnotation{range}(#1)}

\newcommand\rul[1][]{\ensuremath{r_{#1}}\xspace}

\newcommand\head[1]{\opnotation{head}(#1)} %
\newcommand\body[1]{\opnotation{body}(#1)} %
\newcommand\varsOf[1]{\opnotation{Var}(#1)} %
\newcommand\exVarsOf[1]{\opnotation{Var}_\exists(#1)} %
\newcommand\uniVarsOf[1]{\opnotation{Var}_\forall(#1)} %

\newcommand\Inst[1][I]{\ensuremath{\mathcal{#1}}\xspace}
\newcommand\Jnst{\Inst[J]}
\newcommand\Dnst{\Inst[D]}

\newcommand\core[1]{\ensuremath{\mathsf{core}(#1)}\xspace}
\newcommand{\nmrnot}{\text{{\bf not}\,}} %

\newcommand\BCQ{BCQ\xspace}
\newcommand\NBCQ{BNCQ\xspace}
\newcommand\BNCQ{\NBCQ}
\newcommand\NBCQs{\NBCQ{s}\xspace}
\newcommand\BNCQs{\NBCQs}

\newcommand\bodyp[1]{\ensuremath{\opnotation{body}^+(#1)\xspace}}
\newcommand\bodyn[1]{\ensuremath{\opnotation{body}^-(#1)\xspace}}

\newcommand\positive[1]{\ensuremath{{#1}^{+}}\xspace}
\newcommand\negative[1]{\ensuremath{{#1}^{-}}\xspace}

\newcommand\vectornotation[1]{\ensuremath{\mathbf{#1}}}
\newcommand\setnotation[1]{\ensuremath{\mathbf{#1}}}

\newcommand\posrely{\ensuremath{\mathrel{\prec^{+}}}\xspace}
\newcommand\negrely{\ensuremath{\mathrel{\prec^{-}}}\xspace}
\newcommand\restrain{\ensuremath{\mathrel{\prec^{\Box}}}\xspace}

\newcommand\lapr{\ensuremath{\mathrel{\preceq}}\xspace}
\newcommand\napr{\ensuremath{\mathrel{\not\approx}}\xspace}
\newcommand\splitt[2][]{\ensuremath{\mathsf{split}_{#1}(#2)}\xspace}
\newcommand\mergg[2][]{\ensuremath{\mathsf{merge}_{#1}(#2)}\xspace}

\newcommand\tlist[1]{\vectornotation{#1}\xspace}

\newcommand\nullS{\setnotation{N}\xspace}
\newcommand\constS{\setnotation{C}\xspace}
\newcommand\predS{\setnotation{P}\xspace}
\newcommand\varS{\setnotation{V}\xspace}

\RequirePackage{tikz}
\usetikzlibrary{arrows,positioning,petri,decorations,shapes,decorations.markings}
\tikzset{
  every place/.style=
    {
      circle,
      draw,
      thick,
      inner sep=3pt,
      minimum size=7mm
    },
  every transition/.style=
    {
      rectangle,
      draw,
      thick,
      inner sep=3pt,
      minimum size=7mm
    },
  edge/.style=
    {
      ->,
      shorten <=1pt,
      >=stealth',
      semithick
    },
  pre/.style=
    {
      <-,shorten <=.1pt,>=stealth',semithick
    },
  pres/.style=
    {
      <-
    },
  post/.style=
    {
      ->,shorten >=.1pt,>=stealth',semithick
    },
  posts/.style=
    {
      ->
    },
  state/.style=
    {
      circle,draw,semithick,inner sep=.1pt,minimum size=1.5mm,fill=black
    },
  entity/.style=
    {
      rounded corners=3,
      draw,
      semithick,
      inner sep=.5em,
      minimum size=2em
    }
}
\newtheorem{definition}{Definition}
\newtheorem{example}{Example}
\newtheorem{proposition}{Proposition}
\newtheorem{lemma}{Lemma}
\newtheorem{theorem}{Theorem}
\newtheorem{corollary}{Corollary}

\usepackage{comment}
\includecomment{aaai}
\excludecomment{arxiv}
\excludecomment{arxiv1}
\excludecomment{arxiv2}
\newcommand\arXiv{
  \excludecomment{aaai}
  \includecomment{arxiv}
  \includecomment{arxiv1}
}
\arXiv

\begin{aaai}
  \title{Answering Queries with Negation over Existential Rules}
\end{aaai}
\begin{arxiv}
  \title{Answering Queries with Negation over Existential Rules\thanks{This paper is the technical report for the paper with the same title appearing at the 36th AAAI Conference on Artificial Intelligence (AAAI 2022).}}
  \nocopyright
\end{arxiv}

\author {
Stefan Ellmauthaler,
Markus Krötzsch,
Stephan Mennicke
}
\affiliations{
Knowledge-Based Systems Group, TU Dresden, Dresden, Germany\\
firstname.lastname@tu-dresden.de
}

\begin{aaai}
  \newcommand\href[2]{#2\xspace}
\end{aaai}
\begin{arxiv}
  \usepackage[hidelinks]{hyperref}
  
\end{arxiv}

\makeatletter%
\begin{document}

\maketitle

\begin{abstract}
Ontology-based query answering with existential rules is well understood and implemented for
positive queries, in particular conjunctive queries.
For queries with negation, however, there is no agreed-upon semantics or standard implementation.
This problem is unknown for simpler rule languages, such as Datalog, where it is intuitive 
and practical to evaluate negative queries over the least model.
This fails for existential rules, which instead of a single least model
have multiple \emph{universal models} that may not lead to the same results for negative queries.
We therefore propose \emph{universal core models} as a basis for a meaningful (non-monotonic) semantics for queries 
with negation. 
Since cores are hard to compute, we identify syntactic conditions (on rules and queries) under which our core-based semantics
can equivalently be obtained for other universal models, such as those produced by practical chase algorithms.
Finally, we use our findings to propose a semantics for
a broad class of existential rules with negation.
\end{abstract}

\section{Introduction}\label{sec_intro}

Existential rules are a prominent approach in knowledge representation, due to theoretical and practical advances in ontology-based query answering~\cite{BagetLMS09,podsdissem}, but also because of applications in many other domains, such as data exchange and integration~\cite{FaginKMP05}.
Answering conjunctive queries (CQs) over sets of existential rules is often the main goal,
where a CQ is entailed if it is satisfied by all models of the given rules.
This is often implemented using \emph{universal models},
which is a class of models each representing all positive query answers, such that one such model
suffices to answer CQs~\cite{DNR08:corechase}.

However, for queries that ask for the absence of facts (\ie queries incorporating negated atoms), universal models cannot
be used because different universal models may yield different query answers.
This is a major problem, since negation is an important feature in real-world queries, and
a semantic prerequisite for supporting negation in rule bodies, which can be viewed as a recursive
generalisation of negative queries. The severity of this limitation is further aggravated by the fact that
no such problems exist for Datalog, where negative queries can safely be evaluated over the unique least model.
This can be further generalised by considering \emph{stratified negation} in rule bodies, which yields an 
intuitive and implementable non-monotonic semantics \cite{Alice}.
But even this basic form of negation is not safe to use with existential rules.

What makes this problem challenging is that we are looking for a non-monotonic form of negation, where
the absence of positive evidence is sufficient to entail negative information.
Indeed, under the classical first-order semantics of negation, no negated queries are ever entailed by a set of
existential rules. However, it is not immediate when a non-monotonic semantics should allow for
additional, non-classical consequences, and several distinct approaches have been proposed for
existential rules \cite{MKH13:reliances,Baget+14:asp-exists-nmr,GottlobHKL14:ExistRulesWfs,AlvianoMP17:circum-tgds,KR20:cores}.
Our goal for this paper is to provide a semantics that agrees with the widely used and generally accepted
semantics of Datalog with stratified negation for rules without existentials, but which also respects
the classical reading of existential quantifiers as mere statements of existence of certain elements without any commitment
to their identity (in contrast to logic programming approaches that use function terms for referring to distinguished objects).

We therefore propose to evaluate negations with respect to universal models that are \emph{cores},
an algebraic property that, intuitively speaking, ensures that they contain no redundant structures.
Many good practical and theoretical results have been obtained for core models~\cite{FaginKP05:core,DNR08:corechase,CKMOR18:cores,KR20:cores},
but computing cores from arbitrary structures is expensive~\cite{HN92:core}.

\newcommand\xstep{1.6cm}
\newcommand\ystep{1.65cm}
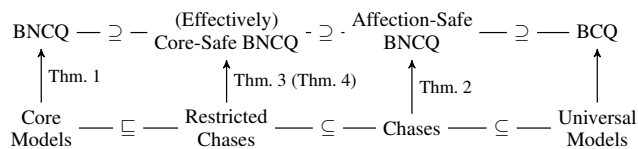
\begin{figure}[tb]
  \centering
  \scalebox{.77}{\begin{tikzpicture}[node distance=.9 and .6]

      \node[align=center] (cores) at (0,0) {Core\\ Models};
      \node[align=center] (subcores) at (\xstep-.10cm,0) {\ensuremath{\sqsubseteq}}
        edge[semithick] (cores);
      \node[align=center] (restricted) at (2*\xstep,0) {Restricted\\ Chases}
        edge[semithick] (subcores);
      \node[align=center] (skolem) at (3*\xstep+.15cm,0) {\ensuremath{\subseteq}}
        edge[semithick] (restricted);
      \node[align=center] (chases) at (4*\xstep,0) {Chases}
        edge[semithick] (skolem);
      \node[align=center] (helper) at (5*\xstep-.05cm,0) {\ensuremath{\subseteq}}
        edge[semithick] (chases);
      \node[align=center] (universals) at (6*\xstep,0) {Universal\\ Models}
        edge[semithick] (helper);

      \node (bcqneg) at (0,\ystep) {\BNCQ}
        edge[pre] node[auto] {\footnotesize Thm.~\ref{prop_coreEntailmentMinimal}} (cores);
      \node[align=center] (h) at (\xstep-.32cm,\ystep) {\ensuremath{\supseteq}}
        edge[semithick] (bcqneg);
      \node[align=center] (coresafe) at (2*\xstep,\ystep) {(Effectively)\\Core-Safe \BNCQ}
        edge[pre] node[auto] {\footnotesize Thm.~\ref{thm:core-safety} (Thm.~\ref{thm:effectively-core-safe})} (restricted)
        edge[semithick] (h);
      \node[align=center] (hh) at (3*\xstep+.12cm,\ystep) {\ensuremath{\supseteq}}
        edge[semithick] (coresafe);
      \node[align=center] (affsafe) at (4*\xstep,\ystep) {Affection-Safe\\\BNCQ}
        edge[semithick] (hh)
        edge[pre] node[auto] {\footnotesize Thm.~\ref{thm:affection-safe}} (chases);
      \node[align=center] (hhh) at (5*\xstep+.30cm,\ystep) {\ensuremath{\supseteq}}
        edge[semithick] (affsafe);
      \node[align=center] (bcq) at (6*\xstep,\ystep) {\BCQ}
        edge[pre] node[auto] {} (universals)
        edge[semithick] (hhh);
    \end{tikzpicture}}
  \caption{Summary of Results}
  \label{fig:summary}
\end{figure}
 We therefore ask if our core-based semantics for negation can also be computed using
other kinds of universal models, which are easier to compute in practice.
It turns out that this is possible if we restrict the shape of the queries that we want to answer,
and we define several classes of ``safe'' queries that can be evaluated over broader classes of models.
Our results are summarised in Figure~\ref{fig:summary}, which illustrates the
inverse relationship between the generality of the query language and the
specificity of the models that one can use to compute them.
The (restricted) \emph{chase} refers to a widely implemented type of reasoning algorithm,
such that the respective query classes could be evaluated in practice~\cite{N+15:RDFoxToolPaper,Benedikt+17:ChaseBench,BSG:Vadalog18,VLog4j2019}.
The symbol $\sqsubseteq$ expresses that core models can be embedded into every restricted chase,
although no such chase is necessarily equal to the core (hence no $\subseteq$).
For any of the given query fragments, we provide concrete syntactic definitions that can be decided
in practice.

Our final contribution is the extension of our approach to existential rules with negations in their bodies.
This is technically more challenging than mere query answering, and even the eager computation of cores
during reasoning cannot guarantee a unique semantics \cite{KR20:cores}.
We focus on cases where rules can be stratified in a certain sense, but our conditions significantly generalise
classical stratification \cite{Alice} and the recent \emph{full stratification} for existential rules \cite{KR20:cores}.
Nevertheless, we can still find a unique ``perfect'' model that can be used in query answering, 
making our approach a valid generalisation of the \emph{perfect core model semantics} \cite{KR20:cores}.

\begin{aaai}
  All proofs are included in the technical report of the paper~\cite{arXivReport2021}.
\end{aaai}
\section{Preliminaries}\label{sec_prelims}

We consider a first-order signature with disjoint sets of \emph{constants} \constS, \emph{(labelled) nulls} \nullS, \emph{variables} \varS, and \emph{predicates} \predS.
A \emph{term} $t$ is an element of $\constS\cup\nullS\cup\varS$.
Lists of terms are denoted $\vec{t}=t_1,t_2,\ldots,t_n$ with $n=|\vec{t}|$, and treated as sets when order is irrelevant.
Each predicate $p\in\predS$ has an \emph{arity} $\arity(p)\in\mathbb{N}$.
An \emph{atom} is an expression $p(\vec{t})$ with $p\in\predS$ and $\arity(p)=|\vec{t}|$.
An atom $p(\vec{t})$ is \emph{ground} if $\vec{t}\subseteq\constS$.
An \emph{interpretation} \Inst is a set of atoms without variables. A database \Dnst is a finite set of ground atoms. %

\paragraph*{Rules and Queries}
An \emph{(existential) rule} $\rul$ is a formula
\begin{align}
\rul = \forall \vec{x}, \vec{y}.\ \varphi[\vec{x}, \vec{y}] \to
\exists \vec{z}.\ \psi[\vec{y}, \vec{z}],\label{eq_rule}
\end{align}
where $\varphi$ and $\psi$ are conjunctions of atoms using only terms from $\constS$ or from the mutually disjoint lists of variables $\vec{x}, \vec{y}, \vec{z}\subseteq\varS$.
We call $\varphi$ the \emph{body} (denoted $\body{\rul}$), $\psi$ the \emph{head} (denoted $\head{\rul}$), and $\vec{y}$ the \emph{frontier} of $\rul$. %
For ease of notation we may treat conjunctions of atoms as sets, and we omit universal quantifiers in rules.
We require that all variables in $\vec{y}$ actually occur in $\varphi$ (\emph{safety}).\footnote{This requirement can be relaxed, but it simplifies presentation.}
A rule is \emph{Datalog} if %
it has no existential quantifiers.

A \emph{normal Boolean conjunctive query} (\BNCQ) is a formula $q = \exists\vec x.\varphi\wedge \psi$, where 
$\varphi$ is a conjunction of atoms with variables from $\vec{x}$, and
$\psi$ is a conjunction of negated atoms $\nmrnot p(\vec{t})$ using only variables that occur in $\varphi$ (\emph{safety}).
We write $q^+$ ($q^-$) for the set of all non-negated atoms in $\varphi$ ($\psi$).
If $q^-=\emptyset$, then $q$ is a \emph{boolean conjunctive query} (BCQ).

\paragraph*{Models and Entailment}
Given a set of atoms $\mathcal{A}$ and an interpretation $\Inst$,
a \emph{homomorphism} $h:\mathcal{A}\to\Inst$ is a function that maps the terms occurring in $\mathcal{A}$
to (the variable-free) terms occurring in $\Inst$,
such that:
(i) for all $c\in\constS$, $h(c)=c$;
(ii) for all $p\in\predS$, $p(\vec{t})\in\mathcal{A}$ only if $p(h(\vec{t}))\in\Inst$, where $h(\vec{t})$ is the list of $h$-images of the terms $\vec{t}$.
If $h$ satisfies (ii) with ``only if'' strengthened to ``iff'', $h$ is a \emph{strong homomorphism}.
An \emph{embedding} is an injective strong homomorphism, and an \emph{isomorphism} is a bijective strong homomorphism (\ie surjective embedding).
We apply homomorphisms to a formula by applying them individually to all of its atoms.

A \emph{match} of a rule $\rul$ in an interpretation $\Inst$ is
a homomorphism $h : \body{\rul}\to\Inst$.
A match $h$ of $\rul$ in $\Inst$ is \emph{satisfied} if there is a homomorphism $h':\head{\rul}\to\Inst$, such that $h\subseteq h'$.
We call $h'$ an \emph{extension of $h$}.
Rule $\rul$ is \emph{satisfied} by $\Inst$, written $\Inst\models\rul$, if every match of $\rul$ in $\Inst$ is satisfied.
A set of rules $\Sigma$ is satisfied by $\Inst$, written $\Inst\models\Sigma$, if $\Inst\models\rul$ for all $\rul\in\Sigma$.
We call \Inst a \emph{model} of $\Sigma$ and database \Dnst if $\Inst\models\Sigma$ and $\Dnst\subseteq\Inst$.
A \BNCQ $q$ is \emph{satisfied} by $\Inst$, written $\Inst\models q$, if
there is a homomorphism $h : q^+\to\Inst$ with $h(q^-)\cap\Inst=\emptyset$.
Note, a \BNCQ $q$ with $q^+\cap q^- \neq\emptyset$ can never be satisfied.
Herein we consider only \BNCQs $q$ that are \emph{non-trivial} in this sense.

\paragraph*{Universal Models and the Chase}
A model $\Inst[U]$ of \Dnst and $\Sigma$ is \emph{universal} if there is a homomorphism $\Inst[U]\to\Inst[M]$ for every model $\Inst[M]$ of $\Dnst$ and $\Sigma$.
Universal models can be computed with \emph{the chase}~\cite{DNR08:corechase}.
In this paper, we will mainly deal with the \emph{restricted chase} (also known as \emph{standard chase}).
\begin{definition}\label{def:restricted-chase}
  Let \Dnst be a database and $\Sigma$ a set of rules.
  A sequence $\Dnst^0,\Dnst^1,\ldots$ is called a \emph{(restricted) chase sequence of \Dnst and $\Sigma$} iff
  \begin{enumerate}
    \item $\Dnst^0 = \Dnst$,
    \item for every $\Dnst^{i+1}$ ($i\geq 0$) there is a rule $\rul\in\Sigma$ of the form \eqref{eq_rule} and a
    match $h$ for \rul in $\Dnst^i$ with
    \begin{enumerate}
      \item $h$ is an unsatisfied match for \rul in $\Dnst^i$ and
      \item there is an extension $h^\star : \head{\rul} \to \Dnst^{i+1}$ of $h$, such that $h^\star(z)$ is a fresh null for each $z\in\vec z$, and
    \end{enumerate}
    \item if $h$ is a match for some rule $\rul\in\Sigma$ in $\Dnst^i$ ($i\geq 0$), then $h$ is satisfied in some $\Dnst^j$ with $j\geq i$ (fairness).
  \end{enumerate}
  $\rul$ is \emph{applicable in $\Dnst^i$} if it has an unsatisfied match in $\Dnst^i$.
  Chase sequence $\Dnst^0, \Dnst^1, \ldots, \Dnst^k$ is called \emph{terminating} if $\Dnst^k$ is a model of $\Sigma$ and $\Dnst$.
  We call $\Dnst^\infty := \bigcup_{i\geq 0} \Dnst^i$ a \emph{chase of $\Sigma$ and \Dnst}.
\end{definition}
Other chase variants, like the Skolem/semi-oblivious or the oblivious chase, mainly differ in the respective definitions of rule applicability (2.~(a) in Definition~\ref{def:restricted-chase}).

\paragraph*{Influence of Existential Quantifiers}
For some of the upcoming notions, it is useful to have rule sets $\Sigma$ where each rule makes use of a distinct set of variables than all the other rules.
We say that $\Sigma$ is \emph{renamed-apart}.
Every rule set can be rewritten to an equivalent renamed-apart rule set.

The \emph{positions} in a predicate $p\in\predS$ are the pairs $\tuple{p,1},\ldots,\tuple{p,\arity(p)}$,
and we will refer to terms \emph{at} a certain position in an atom, set of atoms, or other formula.
$\varsOf{X}$ ($\exVarsOf{X}$/$\uniVarsOf{X}$) denotes the set of all
(existentially quantified/universally quantified) variables in a rule or rule set $X$.
We reproduce the following definition from \citeauthor{KR11:jointacyc}~\shortcite{KR11:jointacyc} to obtain a structure allowing for syntactic reasoning of the influence of existential quantifiers within the chase.

\begin{definition}\label{def_jagraph}
Let $\Sigma$ be a renamed-apart rule set.
For $x\in\varsOf{\Sigma}$, let $\Pi^B_x$ ($\Pi^H_x$) be the sets of all positions
at which $x$ occurs in the body (head) of a (necessarily unique) rule in $\Sigma$.
If $x\in\exVarsOf{\Sigma}$, then
$\Omega_x$ is the smallest set of positions such that
(1) $\Pi^H_x\subseteq\Omega_x$ and
(2) for all $y\in\uniVarsOf{\Sigma}$, $\Pi^B_y\subseteq\Omega_x$ implies $\Pi^H_y\subseteq\Omega_x$.
The set of \emph{jointly affected} positions is $\bigcup_{x\in\exVarsOf{\Sigma}}\Omega_x$.
For $x,y\in\exVarsOf{\Sigma}$,
we write $x\leadsto y$ if the (unique) rule of $y$ has a frontier variable $z$ with $\Pi^B_z\subseteq\Omega_{x}$.
\end{definition}

The relation $x\leadsto y$ states that nulls created for $x$ potentially enable rule applications that create nulls for $y$.
Of particular interest is the transitive and reflexive closure $\leadsto^*$.
Starting from a set of existentially quantified variables, we seek all positions at which a null might occur that is directly or indirectly influenced by a variable in this set.

\begin{definition}\label{def:Vinfluenced}
Let $\Sigma$ be a renamed-apart rule set and let $V\subseteq\exVarsOf{\Sigma}$.
A position $\pi$ is \emph{$V$-influenced} if there are $x\in V$ and $y$ with $x\leadsto^* y$ such that
$\pi\in\Omega_y$.
\end{definition}

Note that \emph{jointly affected} is the same as $\exVarsOf{\Sigma}$-influenced.
Virtually any chase variant (restricted, Skolem, and oblivious) is \emph{compatible} with jointly affected positions, meaning that nulls in chases $\Dnst^\infty$ only occur in $\exVarsOf{\Sigma}$-influenced positions.
\begin{proposition}\label{prop:chase-ja-compatibility}
  For renamed-apart rule set $\Sigma$ and database $\Dnst$ with (Skolem/oblivious/restricted) chase $\Dnst^\infty$, if $p(t_1,\ldots,t_{\arity(p)})\in\Dnst^\infty$ and $t_i\in\nullS$ ($1\leq i\leq \arity(p)$), then $\tuple{p,i}$ is $\exVarsOf{\Sigma}$-influenced.
\end{proposition}
\begin{aaai}
  The proof uses an over-approximation of the chase procedure~\cite{arXivReport2021}.
\end{aaai}
\begin{arxiv}
\begin{proof}
  We prove the claim for a slight adaptation of the restricted chase, removing condition 2.~(a) from Definition~\ref{def:restricted-chase} to obtain the required generality.
  In fact, the resulting chase procedure is an upper bound for all the chase variants.
  Let $\Dnst^\infty$ be such a chase with $p(t_1,\ldots,t_n)\in\Dnst^\infty$ and $t_i\in\nullS$ ($1\leq i\leq n$).
  Then $p(t_1,\ldots,t_n)$ must have been introduced at some chase step $\Dnst^{k}$ by rule $\rul[k]\in\Sigma$ with $p(u_1,\ldots,u_n)\in\head{\rul}$ and match $h_k$ ($h_k^\star(u_i)=t_i$).
  Furthermore, $t_i$ must have been introduced at some chase step $\Dnst^{k_0}$ ($k_0\leq k$) by some rule $\rul[0]$ with existential variable $z\in\exVarsOf{\rul[0]}$ at position $\tuple{q,j}$, and match $h_0$ with $h_0^\star(z)=t_i$.
  Since $t_i$ eventually occurs at position $\tuple{p,i}$, $k-k_0$ is an upper bound for the number of rules different from $\rul$ to have carried $t_i$ from $\Dnst^{k_0}$ to $\Dnst^{k}$.
  If $k=k_0$, then $\rul[0]=\rul$ and $u_i=z$.
  By $z\leadsto^* z$ and $\tuple{p,i}\in\Omega_z$ it follows that $\tuple{p,i}$ is $\exVarsOf{\Sigma}$-influenced.
  If $k>k_0$, then there is a maximal set of atoms $\{ r_1(\vec {v_1}),\ldots,r_m(\vec {v_m}) \} \subseteq \body{\rul}$ with $u_i$ occurring at positions in $\{ \tuple{r_1, j_1},\ldots,\tuple{r_m,j_m} \} = \Pi_{u_i}^B$.
  By an inductive argument, it holds that $\Pi_{u_i}^B\subseteq\Omega_z$, which implies that $\Pi_{u_i}^H\subseteq \Omega_z$.
  Finally, $\tuple{p,i}\in\Pi_{u_i}^H\subseteq \Omega_z$ and, therefore, $\tuple{p,i}$ is jointly affected.
  Thus, $\tuple{p,i}$ is $\exVarsOf{\Sigma}$-influenced.
\end{proof}
\end{arxiv}
\section{Answering \BNCQs on Core Models}\label{sec_cores}

A BCQ is entailed by rule set $\Sigma$ and database $\mathcal{D}$ if it is satisfied
by all models of $\Sigma$ and $\mathcal{D}$. Moreover, every universal model of $\Sigma$ and $\mathcal{D}$
satisfies exactly the BCQs that are entailed in this sense.
Neither condition is appropriate to define entailment of \BNCQs that may use negation.
On the one hand, no such query is satisfied in all models since there is always a model where every BCQ is true.
On the other hand, different universal models do not satisfy the same \BNCQs.
\begin{example}\label{ex:universal-counterexample}
  Take for instance the database $\Dnst = \{ a(1,2), b(2,2) \}$ together with an empty rule set.
  Of course, \Dnst is a universal model, but so is $\Inst[U] := \Dnst \cup \{ a(1,n) \}$ (where $n\in\nullS$).
  As mentioned above, all BCQs evaluate with the same result on \Dnst or \Inst[U].
  However, for \BNCQ $q = \exists x, y .\ a(x,y) \wedge \neg b(y,y)$ we get $\Inst[U]\models q$ while $\Dnst\not\models q$.
\end{example}

Models that are \emph{cores} have been suggested as the appropriate model for defining the semantics of queries that may depend on negative information~\cite{DNR08:corechase,Baget+14:asp-exists-nmr}.
This suggestion is substantiated by the fact that cores are unique up to isomorphism and that isomorphic structures cannot be distinguished by first-order (FO) queries (see, \eg the \emph{Isomorphism Lemma} by \citeauthor{Ebbinghaus1994}~\shortcite{Ebbinghaus1994}).
Indeed, if a universal model is a core, then it satisfies fewer \BNCQs than any other model (see also Theorem~\ref{prop_coreEntailmentMinimal} below),
resulting in the most ``cautious'' notion of non-monotonic entailment.

\begin{definition}
A finite interpretation $\Inst$ is a \emph{core} if every homomorphism $h:\Inst\to\Inst$ is an isomorphism.
A \emph{core model} of a rule set $\Sigma$ and a database $\mathcal{D}$ is a
finite universal model of $\Sigma$ and $\mathcal{D}$ that is a core.
\end{definition}

If $\Sigma$ and $\mathcal{D}$ have any finite universal model, then
they admit a finite core, which is unique up to isomorphism~\cite{DNR08:corechase}.
The situation is more complicated on infinite structures, which admit several
non-equivalent definitions of core~\cite{Bauslaugh95:core-like} and where core models might fail to exist altogether~\cite{CKMOR18:cores}.
We therefore focus on cases with finite models $\Inst[M]$ and their unique cores, which we denote $\core{\Inst[M]}$.
Many conditions have been studied to recognise cases where finite universal models are guaranteed to exist~\cite{CG+13:acyclicity}.
In the rest of the paper we study the notion of \emph{core entailments} for \BNCQs, being those entailments we obtain when evaluating the query on \emph{the} core model.
\begin{definition}\label{def_coreEntailment}
A rule set $\Sigma$ and database $\mathcal{D}$ \emph{core-entail} a \BNCQ $q$, written $\Sigma,\mathcal{D}\models_c q$, if $\Sigma$ and $\mathcal{D}$ have a core model $\mathcal{C}$ satisfying $q$.
\end{definition}

\paragraph*{Important Assumption} Throughout this paper, we only consider 
pairs of rule sets $\Sigma$ and databases $\Dnst$ that have a
(finite) core model. Definition~\ref{def_coreEntailment}
does not apply to other cases.\vspace{5pt}
Again considering Example~\ref{ex:universal-counterexample}, \Dnst is not just any universal model, but the core model.
Hence, $\Sigma,\Dnst\not\models_c q$ for the query $q$ in the example, as one might intuitively expect.
As shown in Example~\ref{ex:universal-counterexample}, such non-entailments are not preserved in arbitrary (universal) models, 
but it turns out that entailments are. 
This is a consequence of the fact that core models are embedded in all universal models in the following sense.
\begin{lemma}\label{lemma_}
  Let $\Sigma$ be a rule set and $\Dnst$ a database with core model $\Inst[C]$ and arbitrary universal model $\Inst[U]$.
  (1) If $h : \Inst[C]\to\Inst[U]$ is a homomorphism, it is an embedding.
  We call $h(\Inst[C])$ the \emph{core instance of $\Inst[U]$} and denote it by $\Inst[U]_c$.
  (2) Every homomorphism $h : \Inst[U]_c\to\Inst[U]_c$ is an isomorphism.
\end{lemma}
\begin{proof}
  On (1), since \Inst[U] is a universal model and \Inst[C] a (core) model of $\Sigma$ and $\Dnst$, there is a homomorphism $h' : \Inst[U]\to\Inst[C]$.
  Note, $h'\circ h : \Inst[C]\to\Inst[C]$ is an isomorphism as $\Inst[C]$ is a core.
  We show that $h$ is an embedding, \ie $h$ is injective and strong.
    For injectivity, let $u$ and $t$ be distinct terms occurring in $\Inst[C]$.
    If $h(u)=h(t)$, it follows that $h'\circ h(u)=h'\circ h(t)$, which contradicts the assumption that $h'\circ h$ is an isomorphism.
    For showing that $h$ is strong, let $\vec u$ be a list of terms in \Inst[C] and $p(h(\vec u))\in\Inst[U]$.
    We need to show that $p(\vec u)\in\Inst[C]$.
    As $h'\circ h$ is an isomorphism on $\Inst[C]$, its inverse is also an isomorphism.
    Hence, $p(h'\circ h(\vec u))\in\Inst[C]$ and $p((h'\circ h)^{-1}\circ (h'\circ h)(\vec u))\in\Inst[C]$ allow for the conclusion $p(\vec u)\in\Inst[C]$.

    Towards (2), there is an isomorphism $i : \Inst[C]\to\Inst[U]_c$ by the argumentation above.
    Let $h : \Inst[U]_c\to\Inst[U]_c$ be a homomorphism.
    Define $j := i^{-1}\circ (h\circ i)$ (\ie $j : \Inst[C]\to\Inst[C]$).
    Since $\Inst[C]$ is a core, $j$ is an isomorphism.
    Hence, $h$ must be an isomorphism.
\end{proof}

Based on the core instance, we can now show that core entailments of \BNCQs carry over to all universal models.

\begin{theorem}\label{prop_coreEntailmentMinimal}
  For every rule set $\Sigma$, database $\Dnst$, and \BNCQ $q$, if $\Sigma,\Dnst\models_c q$, then $\Inst[U]\models q$ for all universal models $\Inst[U]$ of $\Sigma$ and $\Dnst$.
\end{theorem}
\begin{proof}
  Let $\Inst[C]$ be a core model of $\Sigma$ and $\Dnst$.
  The core instance $\Inst[U]_c\subseteq\Inst[U]$ preserves all FO-queries, including \BNCQs.
  We conclude that for all \BNCQ $q$, $\Inst[C]\models q$ implies $\Inst[U]\models q$. %
\end{proof}
In fact, up to minor variations, the core model is the only model with this property:
the only other models for which \BNCQ satisfaction is preserved in all universal models
are those that can be embedded into the core model.
However, it turns out that \emph{some} \BNCQs allow us to consider broader classes
of models for checking core entailment, as we will see in the following sections.

\section{Affection-Safe \BNCQs}\label{sec_affected}

In spite of its appealing semantics, core entailment has the practical disadvantage that
core models are difficult to compute. This difficulty seems unavoidable if we want minimality in
the sense of Theorem~\ref{prop_coreEntailmentMinimal}. However, this does not mean that
core entailment of \BNCQs does always require us to compute the core model. For example,
\BNCQs without negation can equivalently be answered in any universal model. 
We show that similar results can also be obtained for more interesting classes of \BNCQs, and how these
classes can be effectively recognised.

This requires us to look at more specific classes of models.
Indeed, we can readily see that entailment over \emph{all} universal models does often not coincide with core entailment since core-non-entailments are not preserved across universal models:
every \BNCQ $q$ whose positive part is entailed is also true in some universal model,
unless its negative part is a logical consequence (in this case, the non-entailment of $q$ is a first-order logical consequence).
This can be shown by adding redundant structures where $q$ is satisfied, as in Example~\ref{ex:universal-counterexample},
but taking arbitrary rule sets into account.

\begin{proposition}\label{prop_superSafety}
	Consider a rule set $\Sigma$ and database $\Dnst$.
	For every \BNCQ $q$ such that 
	(a) $\Sigma, \Dnst \models \positive q$ and
	(b) $\Sigma,\Dnst \not\models \positive q \wedge a(\vec t)$ for all $a(\vec t)\in\negative q$,
	there is a universal model $\Inst[U]$ of $\Sigma$ and $\Dnst$ with $\Inst[U]\models q$.
\end{proposition}
\begin{proof}
  Define $\Dnst^+ := \Dnst \cup \{ p(\nu(\vec{u})) \mid p(\vec{u})\in\positive q \}$
  with $\nu(t)=t$ for all $t\in\constS$ and $\nu(t)$ a fresh null for all $t\in\varS$.
  Now let $\Inst[U]$ be a chase over $\Sigma$ and $\Dnst^+$ (e.g., some restricted chase).
  By construction, $\Inst[U]$ is a model of $\Sigma$ and $\Dnst$, and $\nu$ is a match
  $\nu:\positive q\to \Inst[U]$. By safety, $\nu$ is defined for $\varsOf{\negative q}$.
  
  Now for any $\alpha=a(\vec t)\in\negative q$, let $\Inst[U]_\alpha$ be a universal model of $\Sigma$ and $\Dnst$ such that
  $\Inst[U]_\alpha\not\models \positive q \wedge \alpha$, which exists due to (b).
  Due to (a), there is a match $\nu':\positive q\to \Inst[U]_\alpha$ such that $\nu'(\alpha)\notin\Inst[U]_\alpha$ (using (b)).
  Since, moreover, $\Dnst\subseteq\Inst[U]_\alpha$, we find that there is a homomorphism $h:\Dnst^+\to \nu'(\positive q)\cup\Dnst$.
  By soundness of the chase, $h$ extends to a homomorphism $h':\Inst[U]\to\Inst[U]_\alpha$.
  Since $\Inst[U]_\alpha$ is universal, this shows that $\Inst[U]$ is universal.
  By construction, $h'(\nu(\alpha))=\nu'(\alpha)\notin\Inst[U]_\alpha$, and hence $\nu(\alpha)\notin\Inst[U]$.
  Applying the same reasoning to all $\alpha\in\negative q$, we find that $\nu$ shows
  $\Inst[U]\models q$.
\end{proof}

Hence, only limited classes of \BNCQs can be evaluated over arbitrary universal models.
This, however, is due to some universal models containing unmotivated clutter
that users might intuitively not expect. In particular, the universal models that 
are computed by chase procedures are not of this form, and can
only produce nulls in jointly affected positions (Proposition~\ref{prop:chase-ja-compatibility}).
Therefore, if all variables in $\negative q$ are bound to some position in $\positive q$ that is not jointly affected,
then these variables can only match constants.

\begin{definition}\label{def_affectionSafety}
  Let $\Sigma$ be a rule set.
  A \BNCQ $q$ is \emph{affection-safe \wrt $\Sigma$} if every variable in $\negative q$ occurs at a position in $\positive q$ that is not jointly affected.
\end{definition}

Since all common (e.g., restricted, Skolem, oblivious) chases are compatible with joint affection,
we can use their outputs to answer affection-safe \BNCQs.

\begin{theorem}\label{thm:affection-safe}
  For every rule set $\Sigma$, database $\Dnst$, and affection-safe \BNCQ $q$, we have $\Sigma,\Dnst\models_c q$ iff $\Inst[M]\models q$ for some\footnote{The restricted chase, e.g., can yield distinct chase results depending on rule application order. Any of them can be used here.}
  chase $\Inst[M]$ computed for $\Sigma$ and $\Dnst$ by a chase procedure compatible with jointly affected positions (Proposition~\ref{prop:chase-ja-compatibility}).
\end{theorem}
\begin{proof}
  The ``only if'' direction follows from Theorem~\ref{prop_coreEntailmentMinimal}.
  For the ``if'' direction, let $\Inst[M]$ be as in the claim.
  By Proposition~\ref{prop:chase-ja-compatibility}, terms $t$ at positions $\tuple{a,i}$ in $\Inst[M]$ that are not jointly affected are constants. %
  Let $\Inst[C]$ be the core model of $\Sigma$ and $\Dnst$. By universality, there are homomorphisms $h_1 : \Inst[M]\to\Inst[C]$ and $h_2 : \Inst[C]\to\Inst[M]$.
  Suppose $\Inst[M]\models q$ by a match $h : \positive q \to \Inst[M]$ with $h(\negative q)\cap\Inst[M] = \emptyset$, but
  $\Sigma,\Dnst\not\models_c q$. 
  Since $h_1\circ h(\positive q) \subseteq\Inst[C]$, this means $h_1\circ h(\negative q) \neq \emptyset$.
  Hence there is $p(\vec t)\in \negative q$ with $p(h_1\circ h(\vec t))\in\Inst[C]\setminus\Inst[M]$,
  and therefore $p(h_2\circ h_1\circ h(\vec t))\in\Inst[M]$.
  As $h(\vec t)$ contains only constants, this implies $p(h(\vec t))\in\Inst[M]$, contradicting our assumption that $\Inst[M]\models q$.
\end{proof}

\section{Core-Safe \BNCQs}\label{sec_collapsing}

Affection safety ensures that chase-based models only entail negative \BNCQ atoms if they match null-free facts,
which always agree across all universal models.
We now relax this requirement based on the recent notion of \emph{restraints} \cite{KR20:cores},
which allow us to identify a larger set of ``safe'' positions to bind to.
The next example, adopted from \citeauthor{AlvianoMP17:circum-tgds}~\shortcite{AlvianoMP17:circum-tgds}, illustrates the problem:
\begin{example}\label{ex:father-core-vs-universal}
  Take as database $\Dnst = \{ p(\text{A}), f(\text{B},\text{A}) \}$
  (mnemonics: \textbf{p}arent, \textbf{f}ather-of, \textbf{e}qual)
  and rules 
  \begin{align}
      f(x,y) & \rightarrow  e(x,x) \label{eq:equali}\\
      p(x) & \rightarrow  \exists y .\ f(y,x)\wedge e(y,y) \label{eq:redu}
   \end{align}
   We can obtain two restricted chases on $\Sigma=\{ \eqref{eq:equali}, \eqref{eq:redu} \}$ and \Dnst: $\Inst[U]_1 = \Dnst \cup \{ e(\text{B},\text{B}) \}$ and $\Inst[U]_2 = \Dnst \cup \{ f(n, \text{A}), e(n,n), e(\text{B},\text{B}) \}$.
  For $\Inst[U]_1$, we only apply \eqref{eq:equali} for match $h_1=\{x\mapsto\text{B}, y\mapsto\text{A}\}$;
  then \eqref{eq:redu} is satisfied.
  For $\Inst[U]_2$, we apply \eqref{eq:redu} for match $h_2=\{x\mapsto\text{A}\}$, and then
  rule \eqref{eq:equali} for match $h_1$.
  The \BNCQ $q = \exists x_1, x_2, y .\ f(x_1, y) \wedge f(x_2, y) \wedge \neg e(x_1,x_2)$ is such that $\Inst[U]_1\not\models q$ and $\Inst[U]_2\models q$.
  Since $\Inst[U]_1$ is the core model of $\Sigma$ and \Dnst, we obtain $\Sigma,\Dnst \not\models_c q$.
\end{example}
Arguably, the construction of $\Inst[U]_2$ should not have applied \eqref{eq:redu} with the extended
match $h_2^\star=\{x\mapsto\text{A}, y\mapsto n\}$, because $\Inst[U]_2$ eventually does not
need $n$ to satisfy \eqref{eq:redu}.
Indeed, $h_2^+=\{x\mapsto\text{A}, y\mapsto \text{B}\}$ would be another way to extend
$h_2$ to satisfy \eqref{eq:redu}.
Functions that remap heads of prior rule applications to alternative structures
in a chase are called \emph{alternative matches}, and they can be shown to occur whenever
a restricted chase fails to produce a core \cite{KR20:cores}.
\begin{definition}\label{def:alternative_matches}
  Let $\Inst_a\subseteq\Inst_b$ be interpretations such that $\Inst_a$ is obtained from applying rule $\rul$ for extended match $h^\star$.
  A homomorphism $h' : h^\star(\head{\rul})\to\Inst_b$ is an \emph{alternative match of $h$} if
  \begin{itemize}
    \item $h'(t)=t$ for all terms $t$ in $h^\star(\body{\rul})$, and
    \item there is a null $n$ in $h^\star(\head{\rul})$ that does not occur in $h'(h^\star(\head{\rul}))$.
  \end{itemize}
\end{definition}

The occurrence of alternative matches in a restricted chase is associated with structures that
do not occur in the core~\cite{KR20:cores}. Such structures could lead to additional \BNCQ matches.
Since alternative matches involve nulls, affection safety can mitigate this by forcing variables in
negative atoms to match constants.
However, it is more general and still safe if we merely restrict to elements that are
not directly or indirectly related to potential alternative matches.
The existence of alternative matches in a real chase is undecidable, but we can (over)approximate
such matches by considering chase-like interactions of pairs of rules. 
Based on this approach of defining \emph{restraints} between rules~\cite{KR20:cores}, we
can identify existential variables that are at risk of producing redundant structures:

\begin{definition}\label{def:restrained_variables}
  A rule $\rul[1]$ \emph{restrains} rule $\rul[2]$, written $\rul[1]\restrain{}\rul[2]$, if there are interpretations 
  $\Inst[I]_a\subseteq\Inst[I]_b$ and a function $h_2$ where
  \begin{enumerate}
  \item $\Inst[I]_b$ is obtained by applying $\rul[1]$ for match $h_1$,
  \item $\Inst[I]_a$ is obtained by applying $\rul[2]$ for match $h_2$,
  \item $h_2$ has an alternative match $h'{:}\; h_2^\star(\head{\rul[2]})\,{\to}\,\Inst_b$, and
  \item $h_2$ has no alternative match $h_2^\star(\head{\rul[2]})\to\Inst_b\setminus h_1^\star(\head{\rul[1]})$.
  \end{enumerate}
  In this situation, a variable $x\in\exVarsOf{\rul[2]}$ is \emph{restrained} if $h_2^\star(x)$ does not occur in
  $h'( h_2^\star(\head{\rul[2]}))$.
  We write $R_\Sigma$ for the set of all restrained variables of a rule set $\Sigma$.
\end{definition}

Variables in $R_\Sigma$ may produce nulls that are not represented in the core model.
Moreover, such nulls may be involved in further rule applications that derive additional
structures that deviate from the core. To find positions of elements that are certain to agree with 
the core, we therefore consider all positions that are influenced by
restrained variables in the sense of Definition~\ref{def:Vinfluenced}.

\begin{definition}\label{def:core-safety}
  Let $\Sigma$ be a rule set.
  A position $\pi$ is \emph{core-safe (\wrt $\Sigma$)} if it is not $R_\Sigma$-influenced.
  A \BNCQ $q$ is \emph{core-safe \wrt $\Sigma$} if every variable in $\negative q$ occurs at a core-safe position in $\positive q$.
\end{definition}

Reconsidering query $q$ of Example~\ref{ex:father-core-vs-universal}, we find that $q$ is not core-safe \wrt the given rule set.
Positions $\tuple{e,1},\tuple{e,2},\tuple{f,1}$ are not core-safe and variables $x_1$ and $x_2$ both occur at $\tuple{f,1}$ in $\positive q$.
This explains why this example admits restricted chase sequences on which the entailment of $q$ disagrees with the core model.
Indeed, for core-safe \BNCQs, this problem can never occur:

\begin{theorem}\label{thm:core-safety}
  For every rule set $\Sigma$, database $\Dnst$, core-safe \BNCQ $q$, and
  restricted chase $\Inst[M]$ of $\Sigma$ and $\Dnst$, we find that 
  $\Sigma,\Dnst\models_c q$ iff $\Inst[M]\models q$.
\end{theorem}

The ``only if'' direction is again clear from Theorem~\ref{prop_coreEntailmentMinimal}.
The proof for the ``if'' direction is given below.
The key insight is that terms at core-safe positions in a chase $\Inst[M]$ must belong
to the core instance $\Inst[M]_c$ of $\Inst[M]$ (\cf Lemma~\ref{lemma_}).
This is a similar situation as for affection-safety,
where we showed that variables at affection-safe positions much be instantiated with constants,
which therefore occur in the core.

\begin{lemma}\label{lemma:coreSafesAreCore}
  For a restricted chase $\Inst[M]$ of rule set $\Sigma$ and database \Dnst, and core instance $\Inst[M]_c$ of $\Inst[M]$,
  if a term $t$ occurs at a core-safe position in $\Inst[M]$, then $t$ occurs in $\Inst[M]_c$.
\end{lemma}
\begin{aaai}
  \begin{proof}[Proof Sketch]
    Let $\Inst[C]$ be a core model of $\Sigma$ and $\Dnst$, which exists since $\Sigma$ and $\Dnst$ have a finite universal model.
    Let $h_1 : \Inst[M]\to\Inst[C]$ be a homomorphism (existence by $\Inst[M]$ being universal) and $h_2 : \Inst[C]\to\Inst[M]_c$ be the isomorphism for which $h_2^{-1}$ agrees with $h_1$, \ie $h_2^{-1}(t)=u$ if $h_1(t)=u$.
    The existence of an isomorphism $i$ is ensured by Lemma~\ref{lemma_}.
    Furthermore, consider the restriction $h'_1 : \Inst[M]_c\to\Inst[C]$ of $h_1$ to $\Inst[M]_c$.
    By Lemma~\ref{lemma_}~(2), $i\circ h'_1 : \Inst[M]_c\to\Inst[M]_c$ is an isomorphism.
    Set $h_2 := (h'_1)^{-1}$, which is an isomorphism since $(i^{-1}\circ i\circ h'_1)=h'_1$ is an isomorphism.

    We subsequently analyse the homomorphism $h':= h_2\circ h_1$.
    With respect to the choice of $h_2$, we show that if $t$ occurs at a core-safe position in $\Inst[M]$, then $h'(t)=t$.
    The rest of the proof is an induction on the chase sequence of $\Inst[M]$ and is included in the full proof included in the technical report~\cite{arXivReport2021}.\qedhere
  \end{proof}
\end{aaai}
\begin{arxiv}
  \begin{proof}
  Let $\Inst[C]$ be a core model of $\Sigma$ and $\Dnst$.
  Let $h_1 : \Inst[M]\to\Inst[C]$ be a homomorphism (existence by $\Inst[M]$ being universal) and $h_2 : \Inst[C]\to\Inst[M]_c$ be the isomorphism for which $h_2^{-1}$ agrees with $h_1$, \ie $h_2^{-1}(t)=u$ if $h_1(t)=u$.
  The existence of an isomorphism $i$ is ensured by Lemma~\ref{lemma_}.
  Furthermore, consider the restriction $h'_1 : \Inst[M]_c\to\Inst[C]$ of $h_1$ to $\Inst[M]_c$.
  By Lemma~\ref{lemma_}~(2), $i\circ h'_1 : \Inst[M]_c\to\Inst[M]_c$ is an isomorphism.
  Set $h_2 := (h'_1)^{-1}$, which is an isomorphism since $(i^{-1}\circ i\circ h'_1)=h'_1$ is an isomorphism.

  We subsequently analyse the homomorphism $h':= h_2\circ h_1$.
  With respect to the choice of $h_2$, we show that if $t$ occurs at a core-safe position in $\Inst[M]$, then $h'(t)=t$.
  By induction on the chase sequence of $\Inst[M]$ (\ie $\Inst[M]^0, \Inst[M]^1, \ldots$).
  \begin{description}
    \item[Base:] Since $\Inst[M]^0=\Dnst\subseteq\Inst[M]_c$, every term $t$ in $\Inst[M]^0$ is a constant and $h'(t)=t$ because $h'$ is homomorphism.
    \item[Hypothesis:] For $k\in\mathbb{N}$, $h'(t)=t$ for every term $t$ in $\Inst[M]^k$ that occurs at a core-safe position.
    \item[Step:] We consider the step from $\Inst[M]^k$ to $\Inst[M]^{k+1}$ by rule $\rul\in\Sigma$ and (unsatisfied) match $h$ with $\Inst[M]^{k+1}=\Inst[M]^k \cup h^\star(\head{\rul})$.
    Let $a(\vec t)\in\head{\rul}$ and $\tuple{a,i}$ be core-safe ($i\in\{ 1, \ldots, |\vec t|\}$).
    We distinguish two cases for $h^\star(t_i)=u$:
    If $u\in\constS$, $h'(u)=u$ because $h'$ is homomorphism.
    If $u\in\nullS$, $t_i\in\varS$ and we distinguish two more cases, namely (a) $t_i$ is a universally quantified (\ie frontier) variable and (b) $t_i$ is an existentially quantified variable of $\rul$.

    In case (a), $t_i$ occurs at a core-safe position in $\body{\rul}$ by Lemma~\ref{lemma:core-safety-in-rules} (1).
    Thus, the induction hypothesis guarantees $h'(u)=u$.

    For case (b), we show that $u\neq h'(u)$ implies an alternative match $\hat{h'}$ of $h$.
    If $h(x)=t$ for a non-frontier body variable $x$ of $\rul$, $\hat{h'}(t)=t$.
    Otherwise, $\hat{h'}(t)=h'(t)$.
    First, we show that $\hat{h'}(t)=t$ for all terms $t$ occurring in $h(\body{\rul})$:
    If $t\in\constS$, the claim follows from $h'$ being a homomorphism.
    If $t\in\nullS$, there must have been (i) a non-frontier variable $x$ with $h(x)=t$ or (ii) a frontier variable $y$ in $\rul$ with $h(y)=t$.
    In case (i), $\hat{h'}(t)=t$ by definition.
    In case (ii), $y$ occurs at a core-safe position in $\body{\rul}$ by Lemma~\ref{lemma:core-safety-in-rules} (2).
    Thus, the induction hypothesis ensures that $\hat{h'}(t)=h'(t)=t$.

    By the choice of $h_2$, if $h'(u)\neq u$, then $u\notin\range{h'}$.
    Thus, $h'(u)$ is an alternative to $u$, that does not occur in $h'(h^\star(\head{\rul}))$.
    $\hat{h'}$ is an alternative match for $h$ in $\Inst[M]$.
    But then $t_i\in R_\Sigma$ and $\tuple{a,i}\in\Omega_{t_i}$, contradicting our assumption that $\tuple{a,i}$ is core-safe.
  \end{description}
  Hence, every term in a core-safe position in $\Inst[M]$ also occurs in $\Inst[M]_c$.
\end{proof}
\end{arxiv}

\begin{proof}[Proof of Theorem~\ref{thm:core-safety}]
  For the remaining ``if'' direction, suppose $\Inst[M]\models q$.
  Then there is a homomorphism $h : \positive q \to \Inst[M]$ with $h(\negative q)\cap\Inst[M]=\emptyset$.
  For core model $\Inst[C]$ of $\Sigma$ and $\Dnst$ and core instance $\Inst[M]_c$ of $\Inst[M]$, let $h_1 : \Inst[M]\to\Inst[C]$ / $h_2 : \Inst[C]\to \Inst[M]_c$ be the respective homomorphism/isomorphism (by Lemma~\ref{lemma_}).

  Let $r(\tlist t)\in\negative q$.
  We show that $r(h(\tlist t))\notin\Inst[M]$ implies $r(h_1(h(\tlist t)))\notin\Inst[C]$.
  Every variable in $\tlist t$ occurs in at least one core-safe position in $\positive q$ (by core-safety of $q$).
  Thus, every term in $\tlist u = h(\tlist t)$ occurs in $\Inst[M]_c$ by Lemma~\ref{lemma:coreSafesAreCore}.
  Since $r(\tlist u)\notin\Inst[M]$ (and $r(\tlist u)\notin\Inst[M]_c$), $r(h_1(\tlist u))\notin\Inst[C]$ because, otherwise, $r(h_2(h_1(\tlist u)))\in\Inst[M]_c$, which contradicts $r(\tlist u)\notin\Inst[M]_c$ (therefore, $\notin\Inst[M]$) because $h_2\circ h_1$ is an isomorphism on $\Inst[M]_c$ (by Lemma~\ref{lemma_}~(2)).
\end{proof}

Core-safe \BNCQs therefore are a significant generalisation of affection-safe \BNCQs, at the cost of
requiring the use of the restricted chase -- or any other correct chase procedure that ensures that its
results are subsets of some restricted chase. In practice, this includes chase implementations that
use specific strategies to decide the order in which rules are applied, e.g., by prioritising
Datalog rules (which never introduce unnecessary nulls) \cite{UKJDC18:VLogSystemDescription}.
Since Definition~\ref{def:restrained_variables} identifies restrained variables under the assumption
that rule $r_1$ might be applied before rule $r_2$, the fact that some rule application orders are
generally impossible (or simply did not happen in a specific run) allows us to consider
more variables \emph{effectively core-safe}.

\begin{definition}
  Let $\Sigma$ be a rule set, $\Dnst$ a database, and $S = \Dnst^0, \Dnst^1, \Dnst^2, \ldots$ a restricted chase sequence with $\Dnst^{i-1} \xrightarrow{\rul[i], h_{i}} \Dnst^{i}$ for $i\geq 1$, $\rul[i]\in\Sigma$ and (unsatisfied) match $h_i$.
  A variable $x\in R_\Sigma$ is \emph{effectively restrained in $S$} if there are $j<k$, such that $x$ occurs in rule $\rul[j]$ and $\rul[k] \restrain \rul[j]$.
  The set of all effectively restrained variables in $S$ is denoted $R_S$.

  A position $\tuple{a, i}$ is \emph{effectively core-safe \wrt $S$} if it is not $R_S$-influenced.
  A BNCQ $q$ is \emph{effectively core-safe \wrt $S$} if every variable $x$ in $\negative{q}$ occurs at an effectively core-safe position in $\positive{q}$.
\end{definition}
Effective core safety marks even more positions as safe for \BNCQs to query for.
\begin{theorem}\label{thm:effectively-core-safe}
  For every rule set $\Sigma$, database $\Dnst$, restricted chase sequence $S=\Dnst^0, \Dnst^1, \Dnst^2, \ldots$ over $\Sigma$ and $\Dnst$
  with chase result $\Inst[M]=\bigcup_{i\geq 0}\Dnst^i$,
  and \BNCQ $q$ that is effectively core-safe for $S$,
  it holds that
  $\Sigma,\Dnst\models_c q$ iff $\Inst[M]\models q$.
\end{theorem}
\begin{proof}[Proof sketch]
  Similar argumentation as for the proof of Theorem~\ref{thm:core-safety}.
  In particular, terms at effectively core-safe positions do occur in the core instance $\Inst[M]_c$ of $\Inst[M]$ (\cf Lemma~\ref{lemma:coreSafesAreCore}).
\end{proof}

Our results put \BNCQ answering under core entailment semantics into reach for practical
implementations. Indeed, chase procedures, including the restricted chase, are supported
by efficient implementations, and the computation of restraints is of the same complexity
as the application of a single rule, namely $\Sigma_{2}^{\mathsf{P}}$-complete \cite{KR20:cores}.
This may at first seem harder than computing the core, which is known to be \textsf{DP}-complete \cite{FaginKP05:core},
but the crucial difference is that the worst-case complexity of core computation refers to the size of the
whole chase (often millions of facts), whereas for rule applications and restraint checking, it
is about the size of a single rule (typically dozens of facts).
This may explain why no general implementation of the core chase is available today.
\section{Rules with Negation} %
\label{sec:rules_with_negation}
We turn our attention to \emph{normal rules}, which admit negated atoms in their bodies
and can naturally be viewed as a recursive generalisation of \BNCQ answering.
Using our previous insights, we first define a chase-based semantics for such rules in cases
where rule bodies are core-safe.
We then generalise this by \emph{stratifying} rule sets in a way that is compatible with
restraints, generalising the notion of \emph{full stratification} for normal existential rules \cite{KR20:cores}.
This yields a unique and well-defined semantics that we call \emph{perfect core semantics}.

A \emph{normal existential rule} is an expression
$\rul = \forall \vec{x}, \vec{y}.\ \varphi[\vec{x}, \vec{y}]\wedge\chi[\vec{x}, \vec{y}]  \to
\exists \vec{z}.\ \psi[\vec{y}, \vec{z}]$, such
that $\exists\vec{x}, \vec{y}.\ \varphi[\vec{x}, \vec{y}]\wedge\chi[\vec{x}, \vec{y}]$ is a \BNCQ,
and $\psi[\vec{y}, \vec{z}]$ is a conjunction of atoms.
We use $\bodyp{\rul}$ and $\bodyn{\rul}$ for the sets of all atoms in $\varphi$ and $\chi$.
A \emph{match} of $\rul$ in an interpretation $\Inst$ is a homomorphism
$h:\bodyp{\rul}\to\Inst$ with $h(\bodyn{\rul})\cap\Inst=\emptyset$.
Other notions are as defined for rules without negation.

The restricted chase procedure of Definition~\ref{def:restricted-chase} can then directly be
applied to normal rules. This does not in general lead to a sound reasoning algorithm, since
a rule might be applicable due to the absence of a negated atom that is inferred
later on in the chase. Chase sequences where this does not happen have been
called \emph{generating} and can be used to define a kind of stable model semantics
for normal existential rules~\cite{Baget+14:asp-exists-nmr}.

A well-known approach to obtain generating chase sequences is \emph{stratification},
where rules are partitioned into a sequence of sets (``strata'') such that
rules in higher strata cannot derive facts that occur negatively in rules of lower strata.
For (normal) existential rules, stratifications have been defined using three
relations between rules:
\emph{positive reliances} $r_1\posrely r_2$ express that $r_1$ might derive facts that allow $r_2$ to
be applied, 
\emph{negative reliances} $r_1\negrely r_2$ express that $r_1$ might derive facts that prevent a
possible application of $r_2$ since they occur in $\bodyn{r_2}$,
and 
\emph{restraints} $r_1\restrain r_2$ are as in Definition~\ref{def:restrained_variables} with the 
additional condition that $h_1$ and $h_2$ are also matches (for normal rules) with respect to $\Inst_b$.
\begin{aaai}
  Full formal definitions of these notions are given in our report \cite{arXivReport2021}.
\end{aaai}
\begin{arxiv}
  Formal definitions of these notions for normal rules are provided in the appendix.
\end{arxiv}

\begin{example}\label{ex_normal_rels}
  Consider the rules
  (mnemonics: \textbf{p}arent, \textbf{f}ather-of, \textbf{m}ale, \textbf{c}hild-of, \textbf{a}dult, \textbf{o}lder-than-3, \textbf{t}ired)
  \begin{align}
    p(x) & \to \exists v.\ f(x,v) \wedge m(v) \label{ex_nl_fm}\\
    f(x,y) & \to m(y)\wedge c(y,x) \label{ex_nl_mc}\\
    f(x,y)\wedge \nmrnot a(x)\wedge\nmrnot o(x) & \to t(y) \label{ex_nl_t}\\
    a(x) &\to o(x) \label{ex_nl_o}
  \end{align}
  We have $\eqref{ex_nl_mc}\restrain \eqref{ex_nl_fm}$, $\eqref{ex_nl_fm}\posrely \eqref{ex_nl_mc}$, and
   $\eqref{ex_nl_fm}\posrely \eqref{ex_nl_t}$.
   There are no other relations, in particular $\eqref{ex_nl_o}\not\negrely\eqref{ex_nl_t}$ since 
   $\eqref{ex_nl_o}$ is only applicable in cases where $\eqref{ex_nl_t}$ is not applicable anyway.
\end{example}

For a set of normal rules $\Sigma$, let $R_\Sigma$ denote the set of restrained variables as of
Definition~\ref{def:restrained_variables}, modified for normal rules as mentioned above.
The set of $R_\Sigma$-influenced and core-safe positions is defined as before, ignoring negated atoms in rules.
Then a rule $r\in\Sigma$ is \emph{core-safe} if every variable in $\bodyn{r}$ occurs on a
core-safe position in $\bodyp{r}$.
A first simple observation highlights a case where the restricted chase can safely be applied 
to normal rules:

\begin{proposition}\label{prop_normal_std_chase}
  Let $\Sigma$ be a set of normal rules such that (1) all rules in $\Sigma$ are core-safe and
  (2) there is no negative reliance $r_1\negrely r_2$ between any rules $r_1,r_2\in\Sigma$.
  Then every restricted chase sequence over $\Sigma$ and any database $\Dnst$ 
  is generating and yields a model of $\Sigma$ and $\Dnst$.
  Moreover, all of these models are homomorphically equivalent.
\end{proposition}

\begin{arxiv}
  The full proof of this result is included in the appendix.
\end{arxiv}
\begin{aaai}
  The full proof is included in our technical report~\cite{arXivReport2021}.
\end{aaai}

\begin{example}\label{ex_normal_std_chase}
Let $\Sigma$ denote the set of rules in Example~\ref{ex_normal_rels}. 
We find that $R_\Sigma=\{v\}$ because of $\eqref{ex_nl_mc}\restrain \eqref{ex_nl_fm}$,
such that $R_\Sigma$-influenced positions are $\tuple{f,2}$, $\tuple{m,1}$, $\tuple{c,1}$, $\tuple{t,1}$.
Hence, \eqref{ex_nl_t} is core-safe, and Proposition~\ref{prop_normal_std_chase} applies.

For a database $\Dnst=\{p(A), f(A,B)\}$, we can obtain the restricted chases
$\Inst[U]_1=\Dnst\cup\{m(B),c(B,A),t(B)\}$ (by applying \eqref{ex_nl_mc} before \eqref{ex_nl_fm}) and
$\Inst[U]_2=\Inst[U]_1\cup\{f(A,n),m(n),c(n,A),t(n)\}$  (by applying \eqref{ex_nl_fm} before \eqref{ex_nl_mc}).
Though distinct, they are homomorphically equivalent and $\Inst[U]_1$ is their unique core.
\end{example}

Example~\ref{ex_normal_std_chase} also illustrates a case that is covered by Proposition~\ref{prop_normal_std_chase} but
is not in scope of the previously defined \emph{full stratification}, which we recall and adapt next.
As opposed to traditional stratifications, the definition of \citeauthor{KR20:cores}~\shortcite{KR20:cores}
effectively allows some rules (esp.\ those that are not the target of any $\negrely$ or $\restrain$)
to appear in multiple strata.

\begin{definition}\label{def:quasi-strati}
For a set $\Sigma$ of normal rules, a list $\mathcal{S} = \tuple{\Sigma_1,\ldots,\Sigma_n}$ with
$\Sigma=\bigcup_{i=1}^n\Sigma_n$ is a \emph{quasi stratification} if, for all rules $r_1\in\Sigma_i$ and $r_2\in\Sigma_j$,
\begin{enumerate}
\item if $r_1 \posrely r_2$ then $i\leq j$,
\item if $r_1 \negrely r_2$ then $i<j$,
\item if $r_1 \restrain r_2$ then $i\leq j$.\label{strat_case_box}
\end{enumerate}
A quasi stratification $\mathcal{S}$ is a \emph{full stratification} if $i<j$ holds in case \eqref{strat_case_box};
it is a \emph{core-safe stratification} if all rules in $\Sigma_k$ are core-safe for $\Sigma_k$
for all $k\in\{1,\ldots,n\}$.
\end{definition}

\begin{example}\label{ex_strats}
The rules of Example~\ref{ex_normal_rels} do not admit a full stratification
due to the cycle $\eqref{ex_nl_mc}\restrain \eqref{ex_nl_fm}\posrely \eqref{ex_nl_mc}$,
but they can be a stratum in a core-safe stratification. We add further rules
  \begin{align}
    f(x,y) & \rightarrow e(y,y) \label{ex_strats_e}\\
    f(x,y_1) & \wedge f(x,y_2) \wedge \nmrnot e(y_1,y_2) \rightarrow d(y_1,y_2) \label{ex_strats_d}
  \end{align}
Then $\eqref{ex_nl_fm}\posrely\eqref{ex_strats_e}$ and  $\eqref{ex_strats_e}\negrely\eqref{ex_strats_d}$.
Rule \eqref{ex_strats_d} is not core-safe in the set of all rules.
A possible core-safe stratification is $\mathcal{S}=\tuple{\{\eqref{ex_nl_fm}, \eqref{ex_nl_mc}, \eqref{ex_nl_t}, \eqref{ex_nl_o}\},\{\eqref{ex_strats_e}\}, \{\eqref{ex_strats_d}\}}$.
\end{example}

Full stratifications have been used to define the \emph{perfect core model}, as the unique
model obtained by conducting a (necessarily generating) chase that proceeds stratum by stratum.
Core-safe stratification is strictly more general, since the stricter condition $i<j$ in
case \eqref{strat_case_box} implies that $R_{\Sigma_k}$ is empty for every stratum $\Sigma_k$, 
so that its rules are core-safe.
A suitable chase procedure for rule sets that are core-safe stratified is given next.

\begin{definition}\label{def:stratified-semi-core-chase}
Let $\Sigma$ be a normal rule set with core-safe stratification
$\mathcal{S} = \tuple{\Sigma_1,\ldots,\Sigma_n}$, and let $\Dnst$ be a database.
The \emph{core-safe chase sequence} for $\mathcal{S}$ and $\Dnst$ is
a list $\Inst[C]^0, \Inst[C]^1,\ldots, \Inst[C]^n$ such that
  \begin{itemize}
  \item $\Inst[C]^0 = \Dnst$, and
  \item for every $i\in\{1,\ldots,n\}$, $\Inst[C]^i$ is the core of
	a restricted chase over $\Sigma_{i}$ and $\Inst[C]^{i-1}$, provided that this core is finite.
  \end{itemize}
 If such a sequence exists, $\Inst[C]^n$ is called the \emph{core-safe chase} of $\Sigma$ and $\Dnst$ \wrt $\mathcal{S}$, and we denote it by $\mathcal{S}(\Dnst)$.
\end{definition}

Note that each stratum $\Sigma_{i}$ satisfies the conditions of Proposition~\ref{prop_normal_std_chase} since
$\mathcal{S}$ is a core-safe stratification. Hence, the required
restricted chase exists and, since it is finite, has a unique core.
In practice, one can ensure the necessary finiteness by using acyclicity conditions that
guarantee chase termination \cite{CG+13:acyclicity}.

\begin{example}\label{ex_core_safe_chase}
Consider the stratification $\mathcal{S}$ of Example~\ref{ex_strats} and the database
$\Dnst=\{p(A), f(A,B)\}$.
The core $\Inst[C]^1$ of the restricted chase over the first stratum was computed as $\Inst[U]_1$ in Example~\ref{ex_normal_std_chase}. 
$\Inst[C]^2$ then is simply $\Inst[C]^0\cup\{e(B,B)\}$, and $\Inst[C]^3=\Inst[C]^2$ is the resulting core-safe chase.
\end{example}

Example~\ref{ex_core_safe_chase} admits other core-safe stratifications, leading to different core-safe chase sequences,
but the final result is the same for all of them.
The main result of this section is that this is a general property of the core-safe chase,
so that we obtain a unique model that provides a semantics of the underlying rule sets.

\begin{theorem}\label{thm:perfectCoresAreUnique}
  For rule set $\Sigma$ and database $\Dnst$ with core-safe stratifications $\mathcal{S}$ and $\mathcal{S}'$, $\mathcal{S}(\Dnst)$ is isomorphic to $\mathcal{S}'(\Dnst)$.
\end{theorem}
\begin{aaai}
  In the proof in our technical report \cite{arXivReport2021}, we compare core-safe stratifications up to simple transformations.
\end{aaai}
\begin{arxiv}
  In the proof (see the appendix), we compare core-safe stratifications up to simple transformations.
\end{arxiv}
In particular, we consider splitting and merging of strata in a stratification to show that the resulting core-safe chases are isomorphic.
The rest of the proof is concerned with showing that all pairs of core-safe stratifications can be transformed into one another by only sequences of splitting and merging operations.
Moreover, we find that core-safe models generalise the perfect core models that are based 
on full stratifications:

\begin{proposition}\label{prop:fully-stratified}
If a rule set is fully stratified and has a finite perfect core model, then the latter is isomorphic to its core-safe model.
\end{proposition}

Together with Theorem~\ref{thm:perfectCoresAreUnique}, the previous result justifies that
we call our semantics for core-safe stratified rule sets the \emph{perfect core semantics}, since
it generalises the eponymous semantics of \citeauthor{KR20:cores}~\shortcite{KR20:cores} without
giving up on the uniqueness of the model.

\section{Discussion and Conclusion}\label{sec_conclusions}

We have investigated how to answer normal Boolean conjunctive queries (\BNCQs) on sets of existential rules,
and we proposed the use of core models as a semantic reference point for this task.
Arguably, cores are both intuitive and mathematically appealing for defining a non-monotonic 
``negation as failure'' semantics, since they satisfy existential rules but at the same time
minimise the amount of inferences and avoid redundancies. Approaches of truth minimisation are
the basis for most non-monotonic semantics, but are much less obvious when giving up the syntactic
Herbrand semantics of logic programming.

Nevertheless, our approach also has limitations.
One of them is the difficulty of computing cores in practice, which we have addressed by
identifying cases where this can be avoided. This leads to practical procedures, which in fact have
already been implemented, though implementers are often not aware that the method is not
sound for arbitrary negative queries or stratified negation in existential rules \cite{VLog4j2019}.
Our core-safe chase still requires certain core constructions (after each stratum), which might not 
be practical. A possible approach to address this would be to investigate in more detail whether the
intermediate non-core structures are truly problematic for evaluating the following rules and queries.

A more general limitation is that cores only behave well if universal models are finite,
whereas rule sets with infinite models may have several distinct universal cores or no core that is
a universal model at all \cite{CKMOR18:cores}. In this sense, the question which semantic reference
point to choose in general remains open.

\section*{Acknowledgments}
This work is partly supported
by Deutsche Forschungsgemeinschaft (DFG, German Research Foundation) in project number 389792660 (TRR 248, \href{https://www.perspicuous-computing.science/}{Center for Perspicuous Systems}),
by the Bundesministerium für Bildung und Forschung (BMBF, Federal Ministry of Education and Research) in the \href{https://www.scads.de}{Center for Scalable Data Analytics and Artificial Intelligence} (ScaDS.AI),
and by the \href{https://cfaed.tu-dresden.de/news}{Center for Advancing Electronics Dresden} (cfaed).

\newpage
\appendix

\begin{arxiv}
  \section{Appendix}
  \label{sec:appendix}
  The following lemma is a technical preliminary to the proof of Lemma~\ref{lemma:coreSafesAreCore}.
  \begin{lemma}\label{lemma:core-safety-in-rules}
    Let $\Sigma$ be a rule set and $\rul\in\Sigma$ with $a(\vec t)\in\head{\rul}$ and $\tuple{a, i}$ ($1\leq i\leq |\vec t|$) being core-safe.
    In a chase step from some $\Inst[M]$ by rule $\rul$ and unsatisfied match $h$:
    \begin{enumerate}
    \item If $y$ is a frontier variable of $\rul$ occurring at $\langle a, i\rangle$ and $h(y)\in\nullS$, then $y$ occurs at a core-safe position in $\body{\rul}$.
    \item If $z$ is an existentially quantified variable of $\rul$ occurring at position $\langle a,i\rangle$, then $z\notin R_\Sigma$ and all frontier variables of $\rul$ occur at core-safe positions in $\body{\rul}$.
    \end{enumerate}
  \end{lemma}
  \begin{proof}
    We prove both items separately.
    \begin{enumerate}
    \item Suppose $y$ occurs only at non-core-safe positions in $\body{\rul}$.
      Then there must be variables $x, z\in\exVarsOf{\Sigma}$ with $x\in R_\Sigma$, $x\leadsto^* z$, and $\Pi_y^B\subseteq \Omega_z$ (by an inductive argument using the fact that $h(y)$ occurs in all positions in $\Pi_y^B$).
      By Definition~\ref{def_jagraph}, $\Pi_y^H\subseteq \Omega_z$, which contradicts the assumption that $\tuple{a,i}$ is a core-safe position.
    \item Suppose $z\in R_\Sigma$, then $z\leadsto^* z$ and $\langle a, i\rangle\in\Omega_z$, contradicting our assumption that $\langle a,i\rangle$ is core-safe.
      If there is a frontier variable $y$ occurring only at non-core-safe positions in $\body{\rul}$, then there are variables $x,v\in\exVarsOf{\Sigma}$ with $x\in R_\Sigma$, $x\leadsto^* v$,  $\Pi_y^B\subseteq \Omega_v$.
      But then $v\leadsto z$ and $\tuple{a,i}$ is not core-safe since $x\leadsto^* z$ and $\tuple{a,i}\in\Omega_z$.\qedhere
    \end{enumerate}
  \end{proof}
For a thorough understanding of the upcoming proofs, we repeat the formal definitions of positive reliance (\posrely), negative reliance (\negrely), restraints (\restrain) for normal existential rules, and fully stratified rule sets by \citeauthor{KR20:cores}~\shortcite{KR20:cores}.
For a normal rule $r = \forall \vec x, \vec y .\ \varphi[\vec x, \vec y] \wedge \chi[\vec x, \vec y] \rightarrow \exists \vec z .\ \psi[\vec y,\vec z]$, by $\positive r$ we denote the existential rule $\forall \vec x, \vec y .\ \varphi[\vec x, \vec y] \rightarrow \exists \vec z .\ \psi[\vec x, \vec y]$.
A match $h$ for rule $\positive r$ is \emph{generating for $r$} on instance $\Inst$ if $p(h(\vec t))\notin\Inst$ for all $p(\vec t)\in\bodyn{r}$.
A chase sequence $\Dnst^0,\Dnst^1,\Dnst^2,\ldots$ is generating if for all $i>0$, $\Dnst^i$ was obtained by applying rule $r$ for match $h$ implies that $h$ is generating for $r$ in $\Dnst^\infty$.

For existential rules $\rul[1]$ and $\rul[2]$ (without negated body atoms), $\rul[1]$ \emph{positively relies on} $\rul[2]$ (denoted $\rul[1]\posrely\rul[2]$) if there are instancen $\Inst_a\subseteq\Inst_b$ and a function $h_2$ such that
\begin{enumerate}
\item $\Inst_b$ is obtained from $\Inst_a$ by applying \rul[1] for match $h_1$,
\item $h_2$ is an unsatisfied match for \rul[2] on $\Inst_b$, and
\item $h_2$ is not a match for \rul[2] on $\Inst_b$. 
\end{enumerate}
For normal rules $r_1$ and $r_2$, $r_1$ positively relies on $r_2$ (also denoted by $r_1 \posrely r_2$) if $\positive{r_1}\posrely\positive{r_2}$.
Similarly, $r_1$ restrains $r_2$ (also denoted $r_1 \restrain r_2$) if $\positive{r_1}\restrain\positive{r_2}$ by instances $\Inst_a\subseteq\Inst_b$ and matches $h_1$ and $h_2$, as required by Definition~\ref{def:restrained_variables}, such that the matches are generating for rules $\rul[1]/\rul[2]$ on $\Inst_b$.

For normal rules $r_1$ and $r_2$, we say that $r_1$ \emph{negatively relies on} $r_2$ (denoted $r_1 \negrely r_2$) if there are instances $\Inst_a\subseteq\Inst_b$ such that
\begin{enumerate}
\item $\Inst_b$ is obtained by applying $\positive{\rul[1]}$ for match $h_1$,
\item $\Inst_a$ is obtained by applying $\positive{\rul[2]}$ for match $h_2$,
\item $h_2$ is not generating for $\rul[2]$ on $\Inst_b$, and
\item $h_2$ is generating for $\rul[2]$ on $\Inst_b\setminus h_1(\head{\rul[1]})$.
\end{enumerate}

\subsection{Proof of Proposition~\ref{prop_normal_std_chase}}
\begin{proof}
  Let $S = \Dnst^0, \Dnst^1, \Dnst^2, \ldots$ be a restricted chase sequence of $\Sigma$ and \Dnst, and let $\Dnst^\infty$ be its respective chase.
  If $\Dnst^\infty$ is not a model of $\Sigma$ and $\Dnst$, there must be a rule $r\in\Sigma$ with an unsatisfied match $h : \body{r} \to \Dnst^\infty$, contradicting fairness (3.\xspace of Definition~\ref{def:restricted-chase}) of the chase procedure.

  Suppose, $S$ was not generating.
  Then there would be an instance $\Dnst^i$ ($i>0$) in $S$ that is obtained by applying a rule $r$ for match $h$ on $\Dnst^{i-1}$, so that $h$ is not generating for $r$ on $\Dnst^\infty$.
  As $h$ is generating for $r$ \wrt $\Dnst^{i-1}$, there must be a first instance $\Dnst^j$ $(j\geq i)$ in $S$ for which $h$ is not generating, obtained from applying some rule $r'$ for match $h'$ to $\Dnst^{j-1}$.
  In particular, $h$ is generating for $r$ \wrt $\Dnst^{j-1}$.
  But that means, $r' \negrely r$, contradicting assumption (2).

  Let $\Inst[M]_1, \Inst[M]_2$ be two chases of $\Sigma$ and $\Dnst$, and $\Inst[M]_1^0, \Inst[M]_1^1,\ldots$ be the chase sequence to obtain $\Inst[M]_1$.
  We show that there is a homomorphism from $\Inst[M]_1^i$ to $\Inst[M]_2$ for all $i\in\{ 0, 1, \ldots \}$, providing $\Inst[M]_1 \to \Inst[M]_2$.
  $\Inst[M]_2 \rightleftarrows \Inst[M]_1$ is obtained by a symmetric argument.
  \begin{description}
  \item[Base:] For $i=0$, we have $\Inst[M]_1^0 = \Dnst \subseteq \Inst[M]_2$.
    The identity on \Dnst, $\mathsf{id}_{\Dnst}$, is a homomorphism. %
  \item[Step:] For $i>0$ with $\Inst[M]_1^i$ obtained from $\Inst[M]_1^{i-1}$ by applying rule $r$ for match $h$, so that there is a homomorphism $h_{i-1} : \Inst[M]_1^{i-1} \to \Inst[M]_2$ (by induction hypothesis).
    We need to extend $h_{i-1}$ to $h_{i} : \Inst[M]_1^{i}\to\Inst[M]_2$.
    Since $h$ is a match for $r$ in $\Inst[M]_1^{i}$ and $h_{i-1}$ is a homomorphism, $h' = h_i \circ h$ is a match for $\bodyp{r}$ in $\Inst[M]_2$.
    If $h'$ is a match for $r$ in $\Inst[M]_2$, then there is an extension $h'^{\star}$ with $h'^{\star}(\head{r})\subseteq\Inst[M]_2$ and we can use $h'^{\star}$ to define $h_{i}$ as follows:
    $$h_{i} := h_{i-1} \cup \{ (h^\star(x), h'^\star(x)) \mid x \in \exVarsOf{r} \}\text.$$
    It remains to be shown that $h'$ is a match for $r$ in $\Inst[M]_2$.
    If $\bodyn{r}=\emptyset$, the claim follows immediately since $\body{r}=\bodyp{r}$.
    If $\bodyn{r}\neq\emptyset$, we need to show that for all $p(\vec x)\in\bodyn{r}$, $p(h'(\vec x))\notin\Inst[M]_2$.

    Towards a contradiction, suppose there is some $p(\vec x)\in\bodyn{r}$ with $p(h'(\vec x))\in\Inst[M]_2$.
    Let us denote $h'(\vec x)$ by $\vec u$ and $h(\vec x)$ by $\vec t$.
    Then $p(\vec u)\in\Dnst$ because $\Sigma$ is $\negrely$-free (by assumption (2)) and there is a $p$-atom in $\bodyn{r}$.
    Furthermore, since $p(\vec t)\notin\Inst[M]_1^{i-1}$ and $h_{i-1}(\vec t)=\vec u$, there is at least one fresh null in $\vec t$ because otherwise, $\vec t = \vec u$ due to the (inductive) construction of $h_{i-1}$.
    In the course of constructing $h_{i-1}$, these fresh nulls must have been introduced by some rule applications, whereas $\Inst[M]_2$ provides satisfying assignments for these nulls with terms from the input instance $\Dnst$.
    Hence, there must be rules in $\Sigma$ that have been applied to add properties to $\vec u$ in order to admit a homomorphism from $\vec t$ to $\vec u$.
    These rules certainly restrain the rules that produced the fresh nulls in $\vec t$, or have been influenced by non-core null insertions.
    Thus, $\vec t$ contains at least one null that does not belong to the core instance of $\Inst[M]_1$.
    But this contradicts the assumption (1), that $\Sigma$ is core-safe, since at least $r$ is not core-safe.
    Hence, the assumption that $p(\vec u)\in \Inst[M]_2$ is wrong and $h'$ is a match for $r$ in $\Inst[M]_2$.
    \qedhere
  \end{description}
\end{proof}

\subsection{Proof of Theorem~\ref{thm:perfectCoresAreUnique}} %
  \label{sec:proof_of_theorem_ref}

  We subsequently abstract from the actual contents of a quasi stratification and treat it as a list of subsets of $\Sigma$ that meets certain requirements, which we assert by binary relations $\lapr$ and $\napr$ over $\Sigma$.
  Thereby $\lapr$ is the transitive and reflexive closure of all positive and negative reliances in $\Sigma$, as well as restraints over $\Sigma$ (\ie $\lapr := (\posrely \cup \negrely \cup \restrain)^*$), whereas $\napr := \negrely$.
  A core-safe quasi stratification in this framework is a sequence of rule subsets $\mathcal{S} = \tuple{\Sigma_1,\Sigma_2,\ldots,\Sigma_k}$ of $\Sigma$ respecting $\lapr$ and $\napr$, such that for all $r_1\in\Sigma_i$ and $r_2 \in \Sigma_j$, (1) $r_1 \lapr r_2$ implies $i\leq j$, and (2) $r_1\napr r_2$ implies $i\neq j$.
  Note, $r_1 \negrely r_2$ implies $r_1 \lapr r_2$ and $r_1 \napr r_2$, implementing the requirement that $r_1$ must occur in a strictly earlier stratum.
  Throughout the rest of this section we call a core-safe quasi stratification just a stratification.
  We may further assume a rule set $\Sigma$ to be given as a structure $\tuple{\Sigma,\lapr,\napr}$.

  \begin{definition}\label{def:standardForm}
    For rule set $\Sigma$, a stratification $\mathcal S = \tuple{\Sigma_1,\ldots,\Sigma_k}$ is in \emph{standard form} if for all $i\in\{ 1, \ldots, k \}$ and rules $r_1, r_2\in\Sigma_i$, $r_1\lapr r_2$.
  \end{definition}

  \begin{proposition}
    If any stratification of a rule set $\Sigma$ exists, then there is one in standard form.
  \end{proposition}
  \begin{proof}
    If we view $\tuple{\Sigma,\lapr}$ as a graph structure, we can identify its strongly connected components, which we subsequently call $\lapr$-SCCs.
    Take all $\lapr$-SCCs as the components of a standard form stratification.
    A stratification exists if every $\lapr$-SCC $\Sigma_i$ is $\napr$-free, which means that there are no rules $r_1, r_2\in \Sigma_i$ with $r_1\napr r_2$~\cite{Alice}.
    The set of all $\lapr$-SCCs of $\Sigma$ forms a partially ordered set (by $\lapr$), from which one topological ordering can be chosen as the required stratification in standard form.
  \end{proof}

  \begin{definition}\label{def:splitAndMerge}
    For rule set $\Sigma$ and stratification $\mathcal S = \tuple{\Sigma_1,\ldots,\Sigma_{i-1},\Sigma_i,\Sigma_{i+1},\ldots,\Sigma_k}$, define
    \begin{align*}
      \splitt[i]{\mathcal{S}, \Sigma^l, \Sigma^r} & := \tuple{\Sigma_1, \ldots, \Sigma_{i-1}, \Sigma^l, \Sigma^r, \Sigma_{i+1}, \ldots, \Sigma_k}
    \end{align*}
    for non-empty subsets $\Sigma^l, \Sigma^r$ of $\Sigma_i$ with $\Sigma_i = \Sigma^l \cup \Sigma^r$, and
    \begin{align*}
      \mergg[i]{\mathcal S} & := \tuple{\Sigma_1,\ldots,\Sigma_{i-1},\Sigma_i\cup\Sigma_{i+1},\ldots,\Sigma_k}\text.
    \end{align*}
    A split/merge is \emph{valid} if the result is a stratification.
  \end{definition}
  Split and merge operations can be concatenated in an operation sequence $\sigma = o_1 ; \ldots ; o_n$ ($n\in\mathbb{N}$).
  The operation of $\sigma$ applied to a stratification $\mathcal{S}$ works as follows:
  First, $o_1$ is applied to $\mathcal{S}$ yielding a stratification $\mathcal{S}_1$.
  Operation $o_i$ for $i>1$ is applied to the result of the previous step (\ie $\mathcal{S}_{i-1}$).
  Note, $\mergg[i]{\splitt[i]{\mathcal S, \Sigma^l, \Sigma^r}} = \mathcal{S}$ and $\splitt[i]{\mergg[i]{\mathcal S}, \Sigma_i, \Sigma_{i+1}} = \mathcal{S}$.
  \begin{proposition}\label{prop:standardFormSplitting}
    For every stratification $\mathcal S$ of rule set $\Sigma$, a stratification $\mathcal{S}'$ in standard form can be reached by a sequence of valid split operations.
  \end{proposition}
  \begin{proof}
    Let $S$ be a $\lapr$-SCC of $\Sigma$.
    Note, there must be a stratum $\Sigma_i$ of $\mathcal{S}$ with $S\subseteq \Sigma_i$.
    This holds for every $\lapr$-SCCs of $\Sigma$.
    Thus, it is sufficient to iterate over all strata and extract their SCCs via splitting, while maintaining the order $\lapr$.
    $\napr$ can never be violated after splitting a stratum.
  \end{proof}

  \begin{proposition}\label{prop:standardFormConversion}
    Let $\mathcal{S}_1$ and $\mathcal{S}_2$ be two standard form stratifications of $\Sigma$.
    Then there is a sequence of valid split/merge operations on $\mathcal{S}_1$ to yield $\mathcal{S}_2$.
  \end{proposition}
  \begin{proof}
    Merge independent strata and split them in reverse order.
    Since both standard form stratifications obey $\lapr$ and $\napr$ on $\Sigma$, this procedure terminates successfully.
  \end{proof}

  \begin{lemma}\label{lemma:stratConversion}
    Let $\mathcal{S}_1$ and $\mathcal{S}_2$ be two stratifications of rule set $\Sigma$.
    Then we can transform $\mathcal{S}_1$ into $\mathcal{S}_2$ by sequences of valid split/merge operations.
  \end{lemma}
  \begin{proof}
    By Proposition~\ref{prop:standardFormSplitting}, we obtain standard form stratifications $\mathcal{S}_1'$ and $\mathcal{S}_2'$ just by valid sequences of split operations $\sigma_1$ and $\sigma_2$.
    By Proposition~\ref{prop:standardFormConversion}, there is a valid operation sequence $\sigma'$, transforming $\mathcal{S}_1'$ into $\mathcal{S}_2'$.
    Thus, we obtain a sequence $\sigma$ of valid operations by $\sigma_1 ; \sigma' ; \sigma_2^{-1}$,
    where $\sigma_2^{-1}$ is the dual of $\sigma_2$, replacing every $\splitt[i]{\ldots}$ by $\mergg[i]{\ldots}$ and reversing the order of the operation sequence.
  \end{proof}
  Hence, the set of all stratifications forms an equivalence class under valid transformations.
  Next, we show that the result of a single stratum is unique (up to isomorphisms).
  Therefore, we first show that the chase on a single stratum produces a unique model up to homomorphisms, from which the claim directly follows (homomorphically equivalent models have the same core up to isomorphisms).
  The following claim generalizes Proposition~\ref{prop_normal_std_chase} in that we assume an input instance $\Inst$ (with or without nulls) that is core, instead of an input database \Dnst.
  \begin{lemma}\label{lemma:stratumUniqueness}
    Let $\Sigma$ be a set of normal rules such that (1) all rules are core-safe \wrt $\Sigma$ and (2) $\Sigma$ is $\napr$-free. %
    Then every restricted chase sequence over $\Sigma$ and any core instance \Inst is generating and a yields a model of $\Sigma$ and $\Dnst$.
    All of these models are equivalent up to homomorphisms.
  \end{lemma}
  \begin{proof}
    The proof is almost the same as for Proposition~\ref{prop_normal_std_chase}, except that in the part where we show that every two chases are homomorphically equivalent, we need to distinguish fresh nulls (\ie the ones that have been produced throughout the chase) from those nulls already present in $\Inst$.
    Otherwise, the exact same argumentation works here as well.
  \end{proof}

  \begin{corollary}\label{cor:uniqueCore}
    Let $\Sigma$ be a set of normal rules adhering to the restriction of Lemma~\ref{lemma:stratumUniqueness} and \Inst a core instance.
    Then cores are unique for all chases over $\Sigma$ and $\Inst$.
  \end{corollary}
  Next we show a general result about valid splitting and merging:
  The result of a split/merge operation on a rule set yields the same core.
  As a notational convention in this proof we abbreviate the fact that an instance \Jnst was obtained by applying rule $r$ for match $h$ on \Inst by $\Inst \xrightarrow{r,h} \Jnst$.
  \begin{lemma}\label{lemma:splitMergeCorrectness}
    Let $\Sigma$ be a set of rules, so that $\mathcal{S}_0 = \tuple{\Sigma}$ is a core-safe quasi stratification. %
    Furthermore, let $\mathcal S = \tuple{\Sigma_1, \Sigma_2}$ be a valid split of $\tuple{\Sigma}$.
    For every core instance $\Inst$, we have $\mathcal{S}_0(\Inst)$ is isomorphic to $\mathcal S(\Inst)$.
  \end{lemma}
  \begin{proof}
    We know already that the chases of $\Sigma$ and $\Inst$ have a unique core by Corollary~\ref{cor:uniqueCore} (due to Lemma~\ref{lemma:stratumUniqueness}).
    We will therefore look at the core-safe chase sequence of $\mathcal{S}$, which is a sequence of the three core models $\Inst[C]_0,\Inst[C]_1,\Inst[C]_3$.
    $\Inst[C]_0 = \Inst$.
    $\Inst[C]_1$ is the core of a restricted chase sequence
    $\Inst^0, \Inst^1 \ldots \Inst^m$
    of $\Sigma_1$ and $\Inst$, where $\Inst^0 = \Inst$, and $\Inst^{i-1} \xrightarrow{r_i, h_i}\Inst^i$ for an unsatisfied match $h_i$ of $r_i\in\Sigma_1$ for $0 < i \geq m$.
    Since $\Inst[C]_1$ is the core of $\Inst_m$, there is a homomorphism $c_1 : \Inst^m \to \Inst[C]_1$.
    Second, we get the chase sequence
    $\Jnst^0, \Jnst^1, \ldots, \Jnst^n$
    of $\Sigma_2$ and $\Inst[C]_1$, where $\Jnst^0 = \Inst[C]_1$, and $\Jnst^{j-1} \xrightarrow{r_{m+j},h_{m+j}} \Jnst^j$ for $r_{m+j}\in\Sigma_2$ and unsatisfied match $h_{m+j}$.
    Analogously to $c_1$ there is a homomorphism $c_2 : \Jnst^n \to \Inst[C]_2$.
    \paragraph{Simulation:} %
    \label{par:simulation}
    We can use the chase sequences above as a starting point for a chase sequence $\Inst[K]^0,\Inst[K]^1,\ldots$ of $\Sigma$ and $\Inst$.
    The first $m$ steps use the rule applications of the first chase sequence.
    In particular we have $\Inst[K]^{i-1} \xrightarrow{r_i,h_i} \Inst[K]^i$ ($0<i\leq m$).
    For every $i>0$, $\Inst[K]^i=\Inst^i$ may not hold because the injection of fresh nulls cannot be controlled.
    However, $\Inst[K]^i$ and $\Inst^i$ are surely isomorphic.
    Therefore, we will not distinguish $\Inst[K]$-instances from $\Inst$-instances and go on with the proof as if  
    $\Inst[K]^0,\ldots,\Inst[K]^m = \Inst^0,\ldots,\Inst^m$.
    Note, $\Inst[K]^m$ is a model of $\Sigma_1$ and \Inst (by Lemma~\ref{lemma:stratumUniqueness}) and since $\tuple{\Sigma}$ and $\tuple{\Sigma_1, \Sigma_2}$ are core-safe quasi stratifications, $\Sigma_1$ is satisfied in all further steps of the chase sequence. 
    Since $\Inst[C]_1\subseteq\Inst^m = \Inst[K]^m$, we can continue with the second chase sequence as follows: $\Inst[K]^{m}\Inst[K]^{m+1}\ldots\Inst[K]^{m+n}$ is the sequence where $\Inst[K]^{m+j-1}\xrightarrow{r_{m+j},h_{m+j}}\Inst[K]^{m+j}$ ($0<j\leq n$).
    $\Inst[K]^0\ldots\Inst[K]^m\Inst[K]^{m+1}\ldots\Inst[K]^{m+n}$ is an initial chase sequence of $\Sigma$ and $\Inst$.
    As for the first chase sequence, we will further assume that $\Jnst^j\subseteq \Inst[K]^{m+j}$ ($0\leq j\leq n$), although the injection of nulls in $\Inst[K]$-instances may differ from those in \Jnst-instances (isomorphism argument).
    Since we did not have a core construction step on $\Inst[K]^m$, there may be unsatisfied matches for rules in $\Sigma$ in $\Inst[K]^{m+n}$, which we will show to refer to non-core structures only.
    \paragraph{Stuttering:} %
    \label{par:stuttering_}
    Let $\Inst[K]^{m+n+1},\Inst[K]^{m+n+2},\ldots$ be a continuation of the initial chase sequence $\Inst[K]^0\ldots\Inst[K]^{m+n}$ of $\Sigma$ and $\Inst$.
    Recall that every chase step subsequent to $\Inst[K]^{m+n}$ stems from a rule application $r\in\Sigma_2$.
    By induction on $i\geq 0$, we show that there is a homomorphism $\Inst[K]^{m+n+i}\to\Inst[C]_2$.
    \begin{description}
    \item[Base:] For $i = 0$, we have $\Inst[K]^{m+n+i}=\Inst[K]^{m+n}$ for which we take
      $$c_0(t) := \left\{ \begin{array}{ccl}
                            c_2(c_1(t)) && \text{if $t$ is a term of $\Inst[C]_1$} \\
                            c_2(t) && \text{otherwise.}
                          \end{array}\right.$$
      As $\Inst[K]^{m+n} = \Inst^m \cup \Jnst^n$, $c_0(\Inst[K]^{m+n}) = \Inst[C]_2$.
    \item[Step:] For $\Inst[K]^{i-1}\xrightarrow{\rul,h} \Inst[K]^i$ ($m+n < i$), let $c : \Inst[K]^{j-1}\to\Inst[C]_2$ be the homomorphism by induction hypothesis.
      By $h' = c\circ h$ we obtain a match for $\rul$ in $\Inst[C]_2$ because $c$ is a homomorphism and $\body{\rul}$ is a core-safe BNCQ (\cf proof of Proposition~\ref{prop_normal_std_chase}).
      $h'$ is already satisfied in $\Inst[C]_2$, which means there is an extension $h'^{\star}$ of $h'$ with $h'^{\star}(\head{\rul})\subseteq\Inst[C]_2$.
      Define $c' = c \cup \{ (h^\star(x), h'^\star(x)) \mid x\in\exVarsOf{\rho} \}$ as the homomorphism $\Inst[K]^{m+n+j}\to\Inst[C]_2$.
    \end{description}
    Finally, since finiteness of $\Inst[C]_2$ is enforced by Definition~\ref{def:stratified-semi-core-chase}, there the chase sequence $\Inst[K]^0,\Inst[K]^1,\ldots$ is terminating at some instance $\Inst[K]^{m+n+p}$ for $p\geq 0$, for which we have $\Inst[K]^{m+n+p}\to\Inst[C]_2$, making $\Inst[C]_2$ the core of $\Inst[K]^{m+n+p}$.
  \end{proof}

  \begin{proof}[Proof of Theorem~\ref{thm:perfectCoresAreUnique}]
    Let $\mathcal{S}$ and $\mathcal{S}'$ be two core-safe quasi stratifications.
    $\mathcal{S}$ can be transformed into $\mathcal{S}'$ by only valid split/merge operations (by Lemma~\ref{lemma:stratConversion}).
    Each one of these transformations produces a stratification with an equivalent core-safe chase ()up to isomorphisms) by Lemma~\ref{lemma:splitMergeCorrectness}.
    By transitivity of equivalence up to isomorphisms, $\mathcal{S}(\Dnst)$ and $\mathcal{S}'(\Dnst)$ are isomorphic core models of $\Sigma$ and $\Dnst$.
  \end{proof}

\end{arxiv}

\end{document}